\documentclass[journal,table,xcdraw]{IEEEtran}
\usepackage[table,xcdraw]{xcolor}
\usepackage{graphicx}
\graphicspath{../jpeg/}
\DeclareGraphicsExtensions{.pdf,.jpeg,.png}
\usepackage[cmex10]{amsmath}
\allowdisplaybreaks[4]
\usepackage{cuted}
\usepackage{stfloats}%
\usepackage{amssymb}
\usepackage{array}
\usepackage{enumitem}
\usepackage{mdwmath}
\usepackage{mdwtab}
\usepackage{eqparbox}
\usepackage{url}
\usepackage{enumerate}
\usepackage{amsfonts}
\usepackage{algorithmic}
\usepackage{multirow}
\usepackage{makecell}
\usepackage{mathtools}
\usepackage{multirow}
\usepackage{graphicx}
\usepackage{tabularray}
\usepackage[ruled,vlined,linesnumbered]{algorithm2e}
\usepackage{nomencl}
\usepackage{cite}
\usepackage{textcomp,booktabs}

\usepackage{mathrsfs}
\usepackage{amsthm}
\theoremstyle{remark}

\usepackage[utf8]{inputenc}

\newtheorem{lemma}{Lemma}
\newtheorem{definition}{Definition}
\newtheorem{assume}{Assumption}

\newtheorem{prop}{Proposition}

\newtheorem{property}{Property}

\SetKwInput{KwInput}{Input}                
\SetKwInput{KwOutput}{Output} 
\makeatletter
\def\underbracex#1#2{\mathop{\vtop{\m@th\ialign{##\crcr
				$\hfil\displaystyle{#2}\hfil$\crcr
				\noalign{\kern3\p@\nointerlineskip}%
				#1\crcr\noalign{\kern3\p@}}}}\limits}

\def\upbracefilla{$\m@th \setbox\z@\hbox{$\braceld$}%
	\bracelu\leaders\vrule \@height\ht\z@ \@depth\z@\hfill 
	\kern\p@\vrule \@width\p@\kern\p@\vrule \@width\p@\kern\p@\vrule \@width\p@
	$}

\def\upbracefillb{$\m@th \setbox\z@\hbox{$\braceld$}%
	\vrule \@width\p@\kern\p@\vrule \@width\p@\kern\p@\vrule \@width\p@\kern\p@
	\leaders\vrule \@height\ht\z@ \@depth\z@\hfill\bracerd
	\braceld\leaders\vrule \@height\ht\z@ \@depth\z@\hfill
	\kern\p@\vrule \@width\p@\kern\p@\vrule \@width\p@\kern\p@\vrule \@width\p@
	$}

\def\upbracefillc{$\m@th \setbox\z@\hbox{$\braceld$}%
	\vrule \@width\p@\kern\p@\vrule \@width\p@\kern\p@\vrule \@width\p@\kern\p@
	\leaders\vrule \@height\ht\z@ \@depth\z@\hfill
	\kern\p@\vrule \@width\p@\kern\p@\vrule \@width\p@\kern\p@\vrule \@width\p@
	$}

\def\upbracefilld{$\m@th \setbox\z@\hbox{$\braceld$}%
	\vrule \@width\p@\kern\p@\vrule \@width\p@\kern\p@\vrule \@width\p@\kern\p@
	\leaders\vrule \@height\ht\z@ \@depth\z@\hfill\braceru$}

\def\upbracefillbd{$\m@th \setbox\z@\hbox{$\braceld$}%
	\vrule \@width\p@\kern\p@\vrule \@width\p@\kern\p@\vrule \@width\p@\kern\p@
	\bracerd\braceld
	\leaders\vrule \@height\ht\z@ \@depth\z@\hfill\braceru$}

\begin{document}
	\title{\huge Goal-Oriented Wireless Communication\\ Resource Allocation for Cyber-Physical Systems}

\author{Cheng Feng,~\IEEEmembership{Student Member,~IEEE,}
	Kedi Zheng,~\IEEEmembership{Member,~IEEE,}
	Yi Wang,~\IEEEmembership{Member,~IEEE,}\\
        Kaibin Huang,~\IEEEmembership{Fellow,~IEEE,}
	Qixin Chen,~\IEEEmembership{Senior Member,~IEEE.}


\thanks{Manuscript received 06 November 2023; revised 15 March 2024 and 15 July 2024; accepted 17 July 2024. This work is supported by the National Natural and Science Foundation of China under Grant U2066205 and Grant 52307103. The work described in this paper was supported in part by the Research Grants Council of the Hong Kong Special Administrative Region, China under a fellowship award (HKU RFS2122-7S04), the Areas of Excellence scheme grant (AoE/E-601/22-R), Collaborative Research Fund (C1009-22G), and the Grant 17212423. Part of the described research work is conducted in the JC STEM Lab of Robotics for Soft Materials funded by The Hong Kong Jockey Club Charities Trust. (\textit{Corresponding author: Cheng Feng}, E-mail: fengc.2019@tsinghua.org.cn).}
\thanks{C. Feng, K. Zheng, and Q. Chen are the Department of Electrical Engineering, Tsinghua University, Beijing 100084, China.} 
\thanks{Y. Wang and K. Huang are with the Department of Electrical and Electronic Engineering, The University of Hong Kong. } 
}

\markboth{SUBMITTED TO IEEE TRANSACTIONS ON WIRELESS COMMUNICATIONS}
{Shell \MakeLowercase{\textit{et al.}}: Bare Demo of IEEEtran.cls for IEEE Journals}
\maketitle

\begin{abstract}
The proliferation of novel industrial applications at the wireless edge, such as smart grids and vehicle networks, demands the advancement of cyber-physical systems (CPSs). The performance of CPSs is closely linked to the last-mile wireless communication networks, which often become bottlenecks due to their inherent limited resources. Current CPS operations often treat wireless communication networks as unpredictable and uncontrollable variables, ignoring the potential adaptability of wireless networks, which results in inefficient and overly conservative CPS operations. Meanwhile, current wireless communications often focus more on throughput and other transmission-related metrics instead of CPS goals. In this study, we introduce the framework of goal-oriented wireless communication resource allocations, accounting for the semantics and significance of data for CPS operation goals. This guarantees optimal CPS performance from a cybernetic standpoint. We formulate a bandwidth allocation problem aimed at maximizing the information utility gain of transmitted data brought to CPS operation goals. Since the goal-oriented bandwidth allocation problem is a large-scale combinational problem, we propose a divide-and-conquer and greedy solution algorithm. The information utility gain is first approximately decomposed into marginal utility information gains and computed in a parallel manner. Subsequently, the bandwidth allocation problem is reformulated as a knapsack problem, which can be further solved greedily with a guaranteed sub-optimality gap. We further demonstrate how our proposed goal-oriented bandwidth allocation algorithm can be applied in four potential CPS applications, including data-driven decision-making, edge learning, federated learning, and distributed optimization. Through simulations, we confirm the effectiveness of our proposed goal-oriented bandwidth allocation framework in meeting CPS goals.
\end{abstract}

\IEEEpeerreviewmaketitle

\section{Introduction}
The wireless edge is witnessing a proliferation of new industrial applications that span areas such as smart grids, vehicle networks, and the Internet of Things. These applications predominantly involve data-centric tasks, including data-driven decision making, prediction, and control~\cite{bi2015wireless}. As a result, they are transitioning into cyber-physical systems (CPSs), whose performance is influenced by both their physical components and wireless communication networks~\cite{humayed2017cyber}.

The surge in data traffic and the extensive access demands in CPSs present significant challenges for wireless communications~\cite{dawy2016toward}. To address that, the communication community is innovating strategies to optimize the management of wireless communication networks. A key advancement is \textit{semantic communication}~\cite{Qiao2021}, which considers the importance and semantics of data, not just symbols in communications~\cite{uysal2022semantic}. In contrast to the conventional Shannon paradigm, semantic communications only transmit necessary information relevant to specific tasks. The transmission of irrelevant information is omitted~\cite{qin2021semantic}. 

Meanwhile, stakeholders in CPS-related domains gradually recognize the crucial role of wireless communication networks~\cite{burg2017wireless}. Delays, reliability, and other factors can affect the CPS operation goals~\cite{Yu2016}. For example, increased delays and reduced reliability can lead to suboptimal decisions and increase system costs. However, these industries often view wireless communications as unpredictable, unmanageable, and potentially adverse factors in CPSs~\cite{zhang2022modeling}. This perspective may be rooted in conventional communication networks, which prioritize data delivery based on transmission-related objectives, without considering their impact on physical systems. However, in reality, the communication network can be customized to accommodate the customized needs of CPS task goals~\cite{Arsham2022}.

An evident research gap remains: How to measure the contribution of data semantics to CPS operation goals, and how to subsequently manage radio resource allocations in assisting CPSs in achieving their goals? 

\subsection{Related Works}
\subsubsection{CPS Operation Considering Effects of Communication}
Using smart grids as a typical example of the CPS, researchers examined the effects of communication delays and losses on system stability and operations. Ref.\cite{coelho2016small} investigated the effects of communication delays on microgrid small signal stability, while Ref.\cite{ZhengLu-423} explored the effects of communication reliability on demand response. Further studies investigated the influence of communications on distributed algorithms, including distributed economic dispatch \cite{Yang2017} and distributed frequency regulation services \cite{Nazari2020}.

Traditionally, these works perceived communication as unpredictable, difficult to manage, and possibly detrimental. To mitigate such adverse effects, decentralized~\cite{Liang2013}, quantized~\cite{Rabbat2005}, and event-triggered algorithms~\cite{liu2017event} were designed. These approaches aimed to reduce data traffic, thus reducing the potential negative consequences of communication networks~\cite{berahas2018balancing}. However, as the saying goes, there is no free lunch. These communication-efficient algorithms, while reducing data traffic, introduced increased complexity. Crucially, communication networks are not always unmanageable or invariably harmful. Our work aims to challenge the stereotype that communication networks in CPSs are unmanageable.

\subsubsection{Semantic, Task-oriented, Goal-oriented, and Utility-Oriented Communications}
The field of semantic communication has attracted significant interest from both academia and industry. Refs.\cite{Luo2022,Yang2023} provided a comprehensive review of semantic communications. Notably, the semantics of the same piece of data can vary based on its intended task~\cite{shi2023task}. Consequently, some researchers have refined semantic communications to be task-oriented and goal-oriented, emphasizing the extraction and transmission of only task-relevant information~\cite{Deniz2023}. Goal-oriented communication is expected to be integrated into future 6G communication standards~\cite{strinati20216g}. 

Many researchers have devoted their efforts to designing appropriate encoders and decoders for goal-oriented communications, moving beyond Shannon's paradigm~\cite{Xie2022,Xie2023}. It was recognized that extracting semantics based solely on human experiences is challenging~\cite{Lan2021}. Typically, deep learning techniques have been employed to extract the underlying semantics of data~\cite{Xie2021}. While machine learning methods are prevalent, some studies have attempted to integrate information theory to improve the interpretability of results~\cite{shao2021learning}. Semantic communications typically aim to minimize the communication rate required for specific tasks by designing encoders and decoders. Here, we focus on a resource allocation approach for multiple devices. By designing a resource allocation policy that differentiates the contributions of devices to tasks, we can also minimize the total communication rate. Additionally, we can address the dual problem, which is to maximize task completion efficiency while adhering to a sum communication rate constraint.

The communication resource allocation problem has long been viewed as a special case of the utility maximization problem \cite{chiang2007layering}. The network utility aims to express the quality of service requirements of various devices while ensuring fairness in resource allocation. Typical forms of utility include logarithmic utility and max-min fairness metrics \cite{srikant2013communication}. Utility-oriented communications have been further extended into the semantic-based communication paradigm, where utility is measured by the semantic rate. Employing various utility measures enables the applications such as video streaming \cite{wang2023utility} and vehicular networks \cite{miao2024utility}. However, the specific forms of utility remain limited and are difficult to determine. This paper explores a more general form of the goal/utility function that can be directly linked to CPS task requirements. 

\subsubsection{Efficient Communications for Learning}
The rise of learning-related tasks has sparked extensive research on communication-efficient learning~\cite{Park2021}. Some studies focused on data compression and quantization to reduce data traffic, similar to techniques used in the communication-efficient distributed design~\cite{Shi2020}. Interestingly, another subset of research developed the importance-aware approach to identify valuable data in learning tasks and reduce unnecessary data communications. This approach is in close alignment with the spirit of goal-oriented and semantic communications. In the context of data sample transmissions,  Ref.\cite{liu2020data} developed a communication rule based on the importance of data to train a model. Many studies investigated learning problems in a distributed setting, mostly referred to federated learning~\cite{niknam2020federated}. Refs.\cite{chen2021communication,amiri2021convergence,shi2020joint} employed various criteria to identify the importance of gradient updates, subsequently scheduling participation and consequently reducing the overall training time. In fact, our analysis demonstrates that learning can be viewed as a specific CPS task aimed at minimizing training loss and maximizing inference accuracy in the shortest possible time. Our framework extends to a broader range of CPS applications and is not limited to learning.

\subsection{Contributions and Organizations}
In this study, we introduce the unified framework for goal-oriented radio resource allocation problems in CPSs. Our objective is to bridge the gap between the communication community and other CPS domains. For communications, we analytically evaluate the contributions of transmitted data and develop a goal-oriented bandwidth allocation approach that smoothly integrates with current systems. For other CPS domains, we aim to show that the communication network can be tailored to improve CPS performance and better fulfill its goals. Our major contributions include:
\begin{itemize}[leftmargin=8pt]
    \item Introduce a unified framework for goal-oriented radio resource management in CPSs. We formulate a wireless Resource Block (RB) Allocation (RBA) problem in orthogonal frequency division multiple access (OFDMA) based wireless communication networks. The RBA problem maximizes the \textit{information utility gain}, instead of throughput-related objectives, to better ensure the realization of the CPS goals under constraints of limited radio resources.
    \item Propose a divide-and-conquer and greedy solution approach to the goal-oriented RBA problem, taking into account its combinatorial complexity and other practical considerations. The information utility gain is divided into the sum of marginal information utility gains, enhancing problem parallelism. The goal-oriented RBA problem is then reformulated as a knapsack problem with limited RB capacity, solved via a greedy method with a sub-optimality gap guarantee and reduced complexity.
    \item Develop application-specific designs of the goal-oriented RBA for four representative CPS applications. The specific approach to calculating the information utility gain for different applications is further detailed. The applications require massive connections or large-volume data transmissions, covering common CPS applications, including decision-making, learning, and optimization.  We use smart grids and vehicle networks as examples of CPS systems to explain specific applications. 
    \item Conduct a simulation to verify the effectiveness of the goal-oriented RBA in CPSs. The performance of the proposed RBA method is compared with throughput maximization and pure data utility RBA policies, in terms of both throughput and CPS goal values. The results highlight the advantages of the proposed goal-oriented RBA framework in guaranteeing CPS goals under radio resource constraints.
\end{itemize}

The remainder of this paper is structured as follows: Section \ref{system} introduces the system model. Section \ref{solution} outlines the solution challenges, develops the solution method, and addresses practical concerns.  Section \ref{applications} details the computation of information utility gains for four applications: data-driven decision-making, edge learning, federated learning, and distributed optimization. Section \ref{case} presents case studies, followed by conclusions in Section \ref{conclusion}.

\textit{Notations:}  Scalars are represented in normal font. Vectors and matrices are represented in bold font. $\mathbf{0}$ represents a zero matrix/vector, $\mathbf{I}$ indicates the identity matrix. $\bigcup_{i=1}^n S_i$ indicates arbitrary unions of set $S_1 \cup S_2 \cup \cdots \cup S_n$.  $|\boldsymbol{x}|$ denotes the cardinality of the set $\boldsymbol{x}$. $\mathbb{E},\mathbb{H},\mathbb{P}$ represents probability distributions, entropy, and mathematical expectations. $\|\boldsymbol{x}\|_2$ and $\|\mathbf{x}\|_F$ indicate vector $\boldsymbol{x}$'s 2 norm and matrix $\mathbf{x}$'s Frobenius norm, respectively.  $\nabla$ and $\nabla^2$ indicate the gradient and the second-order partial derivatives, respectively.  $\left<\boldsymbol{x},\boldsymbol{y}\right> $ denotes the inner product. $\lceil x\rceil$ denotes the round up functions of $x$.

\section{System Model} \label{system}
This section establishes the basic settings of the CPS system and the theoretical foundations for the entire RBA problem.

\subsection{Cyber-physical System Architecture}
\begin{figure}[tp]
    \centering
    \includegraphics[width=0.49\textwidth]{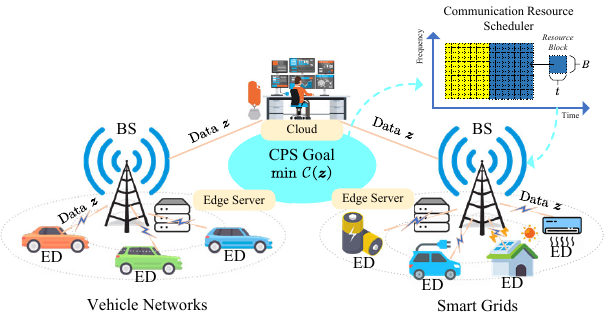}
    \caption{The architecture of the cyber-physical system. The system's ultimate objective is to minimize system goal function $\mathcal{C} \left( \boldsymbol{z} \right)$.}
	\label{fig:intro}
\end{figure}
The architecture of the CPS is illustrated in Fig. \ref{fig:intro}. The physical subsystem comprises various end devices (EDs), labeled as \(j=1,...,J\). In the context of smart grids, EDs might include solar panels, home appliances like air conditioners, and energy storage. For vehicle networks, EDs mainly represent on-road electric vehicles. The cyber-subsystem mostly involves the communication system. In this study, we mainly focus on the uplink communication network.

Within this framework, a CPS operator carries out specific tasks. Usually, CPS operators will process these tasks on nearby (edge) servers or larger cloud servers. The process is as follows: the operator first collects fresh data from the EDs, usually over cellular networks via base stations (BSs). After receiving these data, CPS operators analyze them on the edge server or in the cloud, determine optimal strategies, and if necessary, send commands back to the EDs. 

In the proposed paradigm, the CPS operator no longer views reliability, latency, or other factors as uncontrollable elements. Instead, the CPS operator can control the process of allocation of communication resources. Furthermore, to adaptively support its tasks, the CPS operator allocates communication resources by customizing the system goal function $\mathcal{C} \left( \boldsymbol{z} \right)$ of the allocation problem according to its requirements. The system objective function is contingent on the data $\boldsymbol{z}$ and reflects how the `cyber' side (the communication network) affects the `physical' world (including costs and performance of physical system decisions). Different applications may present different forms of the objective function $\mathcal{C}(\boldsymbol{z})$. 
 
\subsection{Wireless Communication Model}
We assume that orthogonal frequency division multiple access (OFDMA) is utilized for uplink data transmissions. Following LTE (Long-Term Evolution) and 5G cellular communication standards, we use resource block (RB) as the basic unit of resources allocated to EDs in both the frequency and time domains, as shown in Fig. \ref{fig:intro}. For a given RB, labeled as $i$, the data it can deliver (in bits), denoted as $r_{i,j}$, when assigned to ED $j$ is determined by: 
\begin{equation}
	r_{i,j}=tB\log_2\left( 1+\frac{g_j p_j}{B\sigma ^2} \right),
\end{equation}
where $t$ is the time duration of RB $i$. $B$ defines the bandwidth of RB $i$. $g_j$ is the channel gain between the cellular BS and ED $j$. $\sigma ^2$ is the noise power spectral density. $p_j$ is the transmitting power of ED $j$.

During a time period of $T$ with a total of $I$ RBs, the cumulative data bits that an ED $j$ can transmit are given by:
\begin{equation}
	r_j=\sum\nolimits_{i=1}^I{a_{i,j}tB\log _2\left( 1+\frac{g_jp_j}{B\sigma ^2} \right)},
\end{equation}
where $a_{i,j}$ denotes channel allocations, with $a_{i,j}=1$ indicating RB $i$ is allocated to $j$. Otherwise, it equals $0$. In this formula, the small-scale fading effects for physical resource blocks are omitted for simplifications. This enables us to deal with the RB allocation process in a similar manner as the bandwidth allocation problem \cite{chen2021communication,wadu2021joint}.

The focus of this paper is on uplink transmission. Historically, the uplink channel (from EDs to the BSs) often received a smaller share of the communication bandwidth. Consequently, the uplink channel often serves as the bottleneck of the whole network. In most cases, downlink transmissions in CPS applications are broadcast and do not cause communication congestion. 

\subsection{Goal-oriented RBA Problem}
The complete goal-oriented RBA problem can be formulated as Problem (P1) below:
\begin{align}
\max_{a_{i,j},s_j} \,\,&  \mathcal{C} \left( \boldsymbol{z}_{\mathrm{old}} \right)  - \mathcal{C} \left( \boldsymbol{z}_{\mathrm{new}} \right)  \label{eq:objective}
\\
\mathrm{s}.\mathrm{t}.\,\,&\sum\nolimits_{i=1}^I{a_{i,j}tB\log _2\left( 1+\frac{g_jp_j}{\sigma ^2} \right)}\geqslant s_jr_{\min ,j}, \label{eq:rate}
\\
\left( \mathbf{P1} \right)~~~~~&\sum\nolimits_{j=1}^J{a_{i,j}=1,}~~\forall i, \label{eq:allocate1}
\\
&a_{i,j}\in \left\{ 0,1 \right\},s_{j}\in \left\{ 0,1 \right\},~~\forall i,j , \label{eq:allocate2}
\\
\,\,         &\boldsymbol{z}_{\mathrm{new}}=\boldsymbol{z}_{\mathrm{old}}\cup \left\{ \bigcup{s_j\boldsymbol{z}_{j}^{+}} \right\} , \label{eq:data_collection}	
\end{align}
where the decision variables are channel allocations $a_{i,j}$, and ED $j$ transmission permission $s_{j}$ (determining whether ED $j$'s data transmission is allowed). $a_{i,j}$ and $s_{j}$ are both binary variables. $a_{i,j}=1$ indicates RB $i$ is allocated to ED $j$. $s_{j}=1$ implies ED $j$ is allocated sufficient communication resources to facilitate data transmission. The existing dataset before communication is represented by \(\boldsymbol{z}_{\mathrm{old}}\), while \(\boldsymbol{z}_j^+\) denotes the new data generated by ED \(j\), which is not yet known to the CPS operator. \(\boldsymbol{z}_{\mathrm{new}}\) denotes all the data that the CPS operator has acquired after the communication process. The CPS operator decides on ED scheduling $s_j$ and bandwidth allocation $a_{i,j}$. These decisions determine the newly collected data $s_j\boldsymbol{z}_j^+$ from devices, which in turn determine the total data $\boldsymbol{z}_{\mathrm{new}}$ after communications. This ultimately defines the system goal function value $\mathcal{C} \left( \boldsymbol{z}_{\mathrm{new}} \right)$. 

Constraint Eq. \eqref{eq:rate} imposes a restriction to ensure the minimal requirements for reliable data communications: ED \(j\)’s data can be transmitted (\(s_{j}=1\)) if, and only if, the cumulatively allocated channels can assure the minimum reliable data rate \(r_{\min ,j}\) is met. Constraint Eq. \eqref{eq:allocate1} ensures that a single resource block can only be allocated to one ED. Eq. \eqref{eq:data_collection} specifies that the latest dataset $\boldsymbol{z}_{\mathrm{new}}$ that the operators have is the set of unions of the old dataset $\boldsymbol{z}_{\mathrm{old}}$ and freshly transmitted data ${s_{j}\boldsymbol{z}_j^+}$. 

The system goal function, $\mathcal{C} \left( \boldsymbol{z} \right)$, defines a mapping from data to a real number, representing \textit{a measure of information utility}. The objective of Problem (P1) is to maximize the \textit{total information utility gain}, which is the difference between the anticipated system goal function, $\mathcal{C} \left( \boldsymbol{z}_{\mathrm{new}} \right)$, after new data transmission, and $\mathcal{C} \left( \boldsymbol{z}_{\mathrm{old}} \right)$ before data transmission. 

\subsection{System Goal Function $\mathcal{C} \left( \boldsymbol{z} \right)$}
Conventionally, the system goal function, $\mathcal{C} \left( \boldsymbol{z} \right)$, equals the Shannon entropy of data, $\mathbb{H} \left( \boldsymbol{z} \right)$. Accordingly, Problem (P1) seeks to maximize the entropy reduction as much as possible. However, entropy is just one approach to mapping data to a real number, and greater entropy reduction does not always signify greater utility gain. The selection of $\mathcal{C} \left( \boldsymbol{z} \right)$ depends on the requirements of the CPS task and the semantics of the data. 

\textit{Example 1 - Measurement and Estimator}

The goal of the CPS operator is to measure two sensor values, $z_1$ and $z_2$. $z_1$ follows a discrete distribution: $z_1=0$ or $100$, each with a probability of $1/2$. $z_2$ follows a discrete distribution: $z_2=0,1,2,3,4,5$, each with an equal probability of $1/6$. The cost of the CPS operator is defined as the quadratic error between the actual values and their estimators, $\hat{z}_1$ and $\hat{z}_2$, that is, $\mathbb{E}[(\hat{z}_1-z_1)^2+(\hat{z}_2-z_2)^2]$. If the real-time values are unknown to the CPS, the best estimators are the least square estimators: $\hat{z}_1=\mathbb{E}[z_1]$ and $\hat{z}_2=\mathbb{E}[z_2]$. 

Now, the CPS operator can collect real-time measurements of $z_1$ and $z_2$, but unfortunately is limited to collecting just one of them due to communication constraints. Setting the system goal function, $\mathcal{C} \left( \boldsymbol{z} \right)$, as Shannon entropy, the optimal choice becomes collecting data from sensor 2, as knowing $z_2$ leads to a greater reduction in entropy. However, considering the CPS objective function, setting the system goal function as \textit{variance} suggests collecting data from sensor 1 to reduce the CPS cost, $\mathbb{E}[(\hat{z}_1-z_1)^2+(\hat{z}_2-z_2)^2]$, as much as possible. 

\textit{Example 2 -	Image Semantic Data}

The goal of the CPS operator is to determine whether there are pedestrians in the middle of the road using uploaded images, $\boldsymbol{z}$. The CPS goal function, $\mathcal{C} \left( \boldsymbol{z} \right)$, is defined as $-\mathbb{P}[$presence of pedestrians in the middle of the road in images $\boldsymbol{z}]$, with the probability computed through a learning model.

Here, pixels of the images irrelevant to the middle of the road, such as sidewalks and the sky, undoubtedly reduce entropy and provide the operator with a more comprehensive understanding of the road's context. However, this does not enhance inference accuracy as these elements do not add any semantic value to the task the CPS operator intends to accomplish. Irrelevant data can be compressed or even omitted in communications. 

Generally, the task objective $\mathcal{C} \left( \boldsymbol{z} \right)$ — whether explicitly or implicitly aimed at minimizing factors such as decision cost or negative performance — maps data to a real number. Fundamentally, they directly reflect how data values, semantics, and contents contribute to CPS system goals. Therefore, as demonstrated in the two examples, the task objective can directly serve as the system goal function.  Regardless of its formulation, it is recommended that $\mathcal{C} \left( \boldsymbol{z} \right)$ should satisfy the property of \textit{information positivity}:
\begin{property}[Information Positivity]
	Knowing more information can only decrease the value of the system goal function, $\mathcal{C} \left( \boldsymbol{z} \right)$, that is:
	\begin{equation}
		\mathcal{C} \left( \boldsymbol{z} \right) -\mathcal{C} \left( \boldsymbol{z}' \right) \geqslant 0, \forall \boldsymbol{z}\subseteq \boldsymbol{z}'.
	\end{equation}
\end{property}
This property ensures that acquiring more information always benefits the CPS operator. Furthermore, to calculate the information utility gain with low complexity, other necessary properties will be discussed in the following section.

\section{Solution Methods} \label{solution}
\begin{table}[t]
\centering
    \caption{Symbol Description}
\label{tab:my-table}
\resizebox{0.5\textwidth}{!}{%
    \begin{tabular}{cc}
        \toprule
        Symbol        & \multicolumn{1}{c}{Description} \\ \toprule
        \multicolumn{2}{c}{\textit{Sets}}                        \\ \toprule
        $i=1,...,I$   & RB Index                        \\
        $j=1,...,J$   & ED Index                        \\
        $k=1,...,K$   & Iteration Round Index           \\
        $d=1,...,D_j$ & Data Sample Index               \\ \toprule
        \multicolumn{2}{c}{\textit{Variables}}                   \\ \toprule
        $a_{i,j}$     & Channel Allocations             \\
        $s_j$         & Transmission Permissions        \\
          $\boldsymbol{z}^+,\boldsymbol{z}_\mathrm{new}$ &
       \begin{tabular}[c]{@{}l@{}}Transmitted Data, Total Data After Transmissions\end{tabular} \\
          $\boldsymbol{\pi}$                            & \begin{tabular}[c]{@{}l@{}}Decisions and Actions (Data-driven Decision Making)\end{tabular} \\
        $\boldsymbol{\theta },\boldsymbol{\theta}^-,\boldsymbol{\theta}^+$ & \begin{tabular}[c]{@{}l@{}}Model Weights (Edge \& Federated Learning)\end{tabular}   \\ 
          $\boldsymbol{g},\tilde{\boldsymbol{g}}$       & \begin{tabular}[c]{@{}l@{}}Model Gradients (Federated Learning)\end{tabular}                               \\
          $\boldsymbol{\theta}_0,\boldsymbol{\theta}_j$ & \begin{tabular}[c]{@{}l@{}}Decision Variables (Distributed Optimization)\end{tabular} \\
          $\boldsymbol{\lambda}$                        & \begin{tabular}[c]{@{}l@{}}Dual Variables (Distributed Optimization)\end{tabular}                          \\
          \toprule
          \multicolumn{2}{c}{\textit{Parameters}} \\ \toprule
          $t$      & Time Duration of RBs            \\
          $B$      & Bandwidth of RBs                \\
          $g_j$    & Channel Gains                   \\
          $p_j$    & Transmitting Power              \\
          $\sigma$    & Noise Power Spectral Density          \\
          $\boldsymbol{z}_\mathrm{old}$ &
       \begin{tabular}[c]{@{}l@{}}Historical Data\end{tabular} \\  
          $\boldsymbol{\xi}$     & \begin{tabular}[c]{@{}l@{}}Boundary Parameters (Data-driven Decision Making)\end{tabular}                \\
       $\mathsf{z}_j$         & \begin{tabular}[c]{@{}l@{}}Raw Data (Federated Learning and Distributed Optimization)\end{tabular} \\        
          \toprule
   \multicolumn{2}{c}{\textit{Functions}}                                                                                                     \\ \toprule
      $\mathcal{C}(\cdot)$   & CPS Goal                                                                                                     \\
      $h(\cdot)$             & \begin{tabular}[c]{@{}l@{}}Decision Cost (Data-driven Decision Making)\end{tabular}                      \\
    $L(\cdot),\ell(\cdot)$ & \begin{tabular}[c]{@{}l@{}}Loss Function (Edge \& Federated Learning)\end{tabular}               \\
    $\mathcal{L}(\cdot)$   & \begin{tabular}[c]{@{}l@{}}Augmented Lagrangian Function (Distributed Optimization)\end{tabular}                                                                                                                  \\ \bottomrule
    \end{tabular}%
\hspace{0.5em}
}
\end{table}
The goal-oriented RBA problem (P1) is a large-scale combinatorial problem. Directly addressing it may lead to several problems. This section outlines the practical challenges and solution methods for Problem (P1), ensuring its smooth integration into real-world communication systems.
\subsection{Solution Challenges}
 Direct approaches to solving Problem (P1) are hindered by the following computational and practical challenges:
\begin{itemize}[leftmargin=8pt]
\item \textit{Exponential Computation Complexity}: The binary nature of the variable \(s_j\) implies that \(\boldsymbol{z}_{\mathrm{new}} = \boldsymbol{z}_{\mathrm{old}}\cup \left\{ \bigcup{s_j\boldsymbol{z}_{j}^{+}} \right\}\) has an exponential array of potential combinatorial values. This requires the computation of all feasible permutations of \(\mathcal{C} \left( \boldsymbol{z}_{\mathrm{new}} \right)\), corresponding to all possible combinations of \(s_j\). Consequently, this leads to an exponential computational complexity quantified as \(2^{\sum\nolimits_j^{}{| \boldsymbol{z}_{j}^{+} |}}\), where \(| \boldsymbol{z}_{j}^{+} |\) signifies the cardinality of ED \(j\)'s dataset. The resource allocation process itself introduces additional computation overheads and takes additional time. With the complex allocation method, more time is required to determine which resources should be allocated to which devices. Thereby, we should minimize the complexity for the resource allocation process in real communication systems.
\item \textit{Non-Causal Paradox of Future Data}: The objective function of the problem (P1) can be theoretically computed when $\boldsymbol{z}_{\mathrm{new}}$ is known. In practical scenarios, this latest dataset isn't available until communication takes place. Hence, performing RBA based on future data is non-causal and not feasible. 
\item \textit{Complexity and Optimality Gap Trade-off}: The goal-oriented RBA will incur extra overhead and delay for CPSs. If the additional computational overhead to compute allocation decisions is significantly high, it can introduce latency and potentially offset the benefits obtained through the goal-oriented RBA. The practical solution should balance the complexity and the gap to the ideal optimal solution. 
\end{itemize}

\subsection{Divide-and-Conquer Approach}
We start by simplifying the computation of information utility gain. We assume the information utility gain approximately satisfies the submodular property, as shown below: 
\begin{property}[Submodular]
	The negative CPS goal $-\mathcal{C}(\boldsymbol{z})$ is a \textit{submodular function} or can be closely approximated by a \textit{submodular function}. This is defined in the following manner: For every dataset $\boldsymbol{z}$ and $\boldsymbol{z}_1, \boldsymbol{z}_2 \notin \boldsymbol{z}$ such that $\boldsymbol{z}_1 \neq \boldsymbol{z}_2$, we have that $-\mathcal{C}\left(\boldsymbol{z} \cup\left\{\boldsymbol{z}_1\right\}\right)-\mathcal{C}\left(\boldsymbol{z} \cup\left\{\boldsymbol{z}_2\right\}\right) \geqslant -\mathcal{C}\left(\boldsymbol{z} \cup\left\{\boldsymbol{z}_1, \boldsymbol{z}_2\right\}\right)-\mathcal{C}\left(\boldsymbol{z} \right)$. When the inequality is strictly binding, the property becomes the \textit{Additive} property, allowing direct summation of individual utility gains without loss.
\end{property}
Submodular functions are widely used in game theory and machine learning, meaning that obtaining more information is beneficial, but the marginal return diminishes as the dataset grows. With the submodular property, the total information utility gain can be bounded as follows:
\begin{prop}
    When the negative CPS goal satisfies the submodular property, the total information gain can be bounded by the sum of individual information utility gain $\Delta _j$, shown as below:
    \begin{equation}
    \begin{aligned}
    \,\, &\mathcal{C} \left( \boldsymbol{z}_{\mathrm{old}} \right) -\mathcal{C} \left( \boldsymbol{z}_{\mathrm{old}}\cup \left\{ \bigcup{s_j\boldsymbol{z}_{j}^{+}} \right\} \right) 
    \\
    \leqslant &\sum\nolimits_j^{}{\mathcal{C} \left( \boldsymbol{z}_{\mathrm{old}} \right) -\mathcal{C} \left( \boldsymbol{z}_{\mathrm{old}}\cup \left\{ s_j\boldsymbol{z}_{j}^{+} \right\} \right)}
    \\
    \left( \mathbf{P2} \right)~~=&\sum\nolimits_j^{}{s_j\underset{\Delta _j}{\underbrace{\left( \mathcal{C} \left( \boldsymbol{z}_{\mathrm{old}} \right) -\mathcal{C} \left( \boldsymbol{z}_{\mathrm{old}}\cup \left\{ \boldsymbol{z}_{j}^{+} \right\} \right) \right) }}}.
    \end{aligned}
    \end{equation}
\end{prop}

Instead of maximizing the total information utility gain, we maximize its upper bound and replace the objective function as (P2). This approach allows independent computation of the marginal information utility gain $\Delta _j$ for each device. Consequently, this allows parallel computing, reducing the complexity to a linear complexity characterized by $\sum\nolimits_j^{}{| \boldsymbol{z}_{j}^{+} |}$. 

If the system goal function does not satisfy either submodular property or additive property, the system goal function can be replaced with a 'surrogate objective', such as its upper or lower bound, which satisfies the \textit{Submodular} or \textit{Additive} property. Concrete examples will be provided in Section \ref{applications}.

\subsection{Causality Problem}
There are two methods to solve non-causal challenges.

\subsubsection{Large Payload Size Case}
Sometimes, the volume of data of $\boldsymbol{z}^+$ to be transmitted is large. For example, images and deep learning model gradients often have large volumes. In contrast, the marginal information utility comprises just a single numerical value. To manage this, EDs can first send the value of the marginal information utility gain to the server. The server operator then determines which EDs should transmit their full data. By using this method, the operator can address (P2) based on the provided utility values. 

Under these circumstances, extra communication overhead is required to transmit the value of the marginal information utility gain. Given that the volume is significantly smaller than $\boldsymbol{z}^+$, the additional communication overhead can be ignored. Furthermore, if feedback from the CPS operator to EDs is necessary to compute the marginal information utility gain, additional downlink channels are required. These downlink transmissions are typically broadcast messages and will not incur too much extra communication overhead. 

\subsubsection{Small Payload Size Case}
When the data payload is minimal, sending extra data can lead to added congestion in the wireless communication network. Under such circumstances, a practical solution involves sampling fresh data, $\boldsymbol{z}_+$, using historical data, $\boldsymbol{z}_{\mathrm{old}}$, under the assumption that $\boldsymbol{z}_+$ has the same probabilistic distribution as $\boldsymbol{z}_{\mathrm{old}}$. This method leads us to a new optimization problem (P3) where the information utility gain is switched to the expected marginal information utility gain:
\begin{equation}
    \begin{aligned}
        \left( \mathbf{P3} \right)~~~\max_{a_{i,j},s_j}\,\,\sum\nolimits_{j}^{}{\mathbb{E} _{\boldsymbol{z}_{\mathrm{old}}}\left[ s_j\Delta _j \right]},
    \end{aligned}
\end{equation}
where $\mathbb{E} _{\boldsymbol{z}_{\mathrm{old}}}[\cdot]$ represents the expected value based on old data, providing a prediction of potential information gain in upcoming data. This method tends to perform well when there is a strong temporal correlation between consecutive transmitted data.

\subsection{Greedy Solutions and Sub-optimality Guarantee}
Upon calculating the marginal information utility gain, problems (P2-P3) can be converted to the knapsack problem, a well-known combinatorial problem. The communication channel can be converted to a knapsack with a limited capacity of $J$ units. ED $j$ will utilize $w_j$ channel units if selected for the `knapsack', where $w_j$ represents the minimum channel units required to meet the constraint in Eq. \eqref{eq:rate}:
\begin{equation}
	w_j=\lceil \frac{r_{\min ,j}}{Bt\log_2\left( 1+\frac{g_jp_j}{B\sigma ^2} \right)} \rceil.
\end{equation}
Here, the ceiling function is preferred over the floor function to ensure the robustness of the resource allocation process. The benefit derived from ED $j$ corresponds to its marginal information utility gain, which is $\Delta _{j}\geqslant 0$. In summary, (P2-P3) can be transformed into the following problem (P4):
\begin{equation}
\begin{aligned}
	\left( \mathbf{P}\mathbf{4} \right)~~~~\max_{s_j}~~&\sum\nolimits_j^{}{\Delta _js_j}
	\\
	 \mathrm{s}.\mathrm{t}.~~&\sum\nolimits_j^{}{w_js_j\leqslant J},~s_j\in \left\{ 0,1 \right\} .
\end{aligned}
\end{equation}

The knapsack problem is widely recognized as NP-hard, resulting in significant computational overhead. To address this, we employ a greedy algorithm, depicted in Algorithm \ref{al:greedy_algorithm}. Despite its greedy nature, the algorithm's suboptimal loss can be bounded in relation to the theoretical optimum of (P4):
\begin{prop}
If $w_j \leqslant \eta J$  for every ED $j$, with $\eta \in\left(0, \frac{1}{2}\right]$, then the solution computed in Algorithm \ref{al:greedy_algorithm} guarantees at least $1-\eta$ of the optimal solution of (P2)-(P3).
\end{prop}
The proof is provided in Section 2.2 of Ref.\cite{roughgarden2016twenty}. The assumption $w_j \leqslant \eta J$ indicates that a single ED will not consume too many RBs. The most complex part of the greedy solution method is to sort the EDs with respect to $\frac{\Delta_j}{w_j}$. Utilizing an efficient sorting algorithm results in a time complexity of $O(J\mathrm{log}J)$, allowing for swift execution on edge servers.

In summary, with all these techniques, the complete goal-oriented RBA process is shown in Algorithms \ref{al:embb} and \ref{al:mmtc}, for large-size payload and small-size payload, respectively. The entire process is designed to run in parallel. Importantly, the algorithm is tailored to seamlessly integrate with today's cellular communication protocols. Currently, channel allocations are based on channel qualities $g_j$ detected using pilot frequencies. 
\begin{algorithm}[t] 
	\caption{Hybrid Allocation Rule}\label{al:greedy_algorithm}
	\SetAlgoLined
	\KwInput {$w_j$,$\Delta _j$}
	\textbf{Step~1:} Sort and re-index the EDs so that: 
	$$
	\frac{\Delta_1}{w_1} \geq \frac{\Delta_2}{w_2} \geq \cdots \geq \frac{\Delta_J}{w_J} .
	$$\\
	\textbf{Step~2:} Pick EDs in this order until the channel budget is used up, and then halt.\\
	\KwResult {$j\in \{ j:s_j=1 \}$}
\end{algorithm}
\begin{algorithm}[t] 
	\caption{Goal-oriented RBA (Large Payloads)}\label{al:embb}
	\SetAlgoLined
	\textbf{Parallel~Thread~1:} Using pilot frequencies to identify channel gain $g_j$, and compute $w_j$.\\
	\textbf{Parallel~Thread~2:} \\
	\textit{Information Utility Update}: All EDs upload $\Delta _j$.\\
	\textit{Channel Allocations}: The operator employs Algorithm \ref{al:greedy_algorithm} to allocate channels.\\	
	\textit{Application}: Execute system applications using new data $\boldsymbol{z}_{\mathrm{new}}=\boldsymbol{z}_{\mathrm{old}}\cup \left\{ \bigcup{s_j\boldsymbol{z}_{j}^{+}} \right\}$ into database.\\
	\textit{Update}: Update dataset $\boldsymbol{z}_{\mathrm{old}} \gets \boldsymbol{z}_{\mathrm{new}}$.
\end{algorithm}
\begin{algorithm}[t] 
	\caption{Goal-oriented RBA (Small Payloads)}\label{al:mmtc}
	\SetAlgoLined
	\textbf{Data:} {$\boldsymbol{z}_{\mathrm{old}}$}\\
	\textbf{Parallel~Thread~1:} Using pilot frequencies to identify channel gain $g_j$, and compute $w_j$.\\
	\textbf{Parallel~Thread~2:} Using existing dataset to compute $\Delta _j$ for each $j$.\\
	\textbf{Parallel~Thread~3:}\\
	\textit{Channel Allocations}: Employ Algorithm \ref{al:greedy_algorithm} to allocate channels.\\
	\textit{Application}: Execute system applications using new data $\boldsymbol{z}_{\mathrm{new}}=\boldsymbol{z}_{\mathrm{old}}\cup \left\{ \bigcup{s_j\boldsymbol{z}_{j}^{+}} \right\}$ into database.\\
	\textit{Update}: Update dataset $\boldsymbol{z}_{\mathrm{old}} \gets \boldsymbol{z}_{\mathrm{new}}$.
\end{algorithm}

\section{Application-Specific Designs} \label{applications}
\begin{table*}[t]
	\centering
	\caption{Application Scenario Summary}
	\label{tab:summary}
	\begin{tblr}{
			width = \linewidth,
			colspec = {Q[80]Q[100]Q[120]Q[230]Q[200]Q[190]},
			row{1} = {c},
			cell{1}{1} = {c=2}{0.138\linewidth},
			cell{2}{1} = {r=2}{c},
			cell{4}{1} = {r=2}{c},
			vline{2} = {2,3,4,5}{},
			hline{1-2,4,6} = {-}{},
			hline{3,5} = {2-6}{},
		}
		Transmission Data $\boldsymbol{z}^+$ Type &  & Application~ & {System Goal \\$\mathcal{C} \left( \boldsymbol{z} \right)$} & {Information Utility Gain \\$\mathcal{C} \left( \boldsymbol{z}_{\mathrm{old}} \right) - \mathcal{C} \left( \boldsymbol{z}_{\mathrm{new}} \right)$} & {Approximate Marginal\\Information Utility Gain $\Delta _j$}\\
		{Raw Data \\Transmission} & {Sensor Data} & {Data-driven \\Decision Making} & {Decision Cost\\ $\min_{\boldsymbol{\pi }\in \Omega} \max_{\boldsymbol{\xi }\in \Xi (\boldsymbol{z})} \,\,h(\boldsymbol{\xi};\boldsymbol{\pi })$} & {Decision Cost Reduction \\(Proposition \ref{data})} & {Marginal Cost Reduction \\(Proposition \ref{data_m})}\\
		& Data Samples & {Edge \\Learning} & {Training Loss\\$\min_{\boldsymbol{\theta }} \,L(\boldsymbol{z};\boldsymbol{\theta })$} & {Training Loss Reduction \\(Proposition \ref{edge})} & {Individual Inference Error\\ (Proposition \ref{edge_m})} \\
		{Intermediate Data\\Transmission} & Gradients & Federated Learning & {Training Loss~~\\$\min_{\boldsymbol{\theta }} \,L(\mathsf{z};\boldsymbol{\theta })$} & {Loss Reduction via Gradient \\Descent (Proposition \ref{federated})} & {Individual Gradient Norm\\ (Proposition \ref{federated_m})}\\
		& {Primal \& \\Dual Variables} & {Distributed\\Optimization} & {Loss \& Cost\\$\min_{\boldsymbol{\theta }_j=\boldsymbol{\theta }_0} \,\sum\nolimits_j{L_j(\boldsymbol{\theta }_j;\mathsf{z}_j)}$} & {Lagrangian Reduction via \\ADMM (Proposition \ref{distributed})} & {Primal Variable Changes\\ (Proposition \ref{distributed_m})}
	\end{tblr}
\end{table*}
To clarify the implementation of the goal-oriented RBA within CPSs, we will refer to four typical applications, as outlined in Table \ref{tab:summary} and Fig. \ref{fig:commu}.

\subsection{Goal-oriented RBA for Raw Data Transmission}
\begin{figure*}[t]
	\centering
	\includegraphics[width=0.95\textwidth]{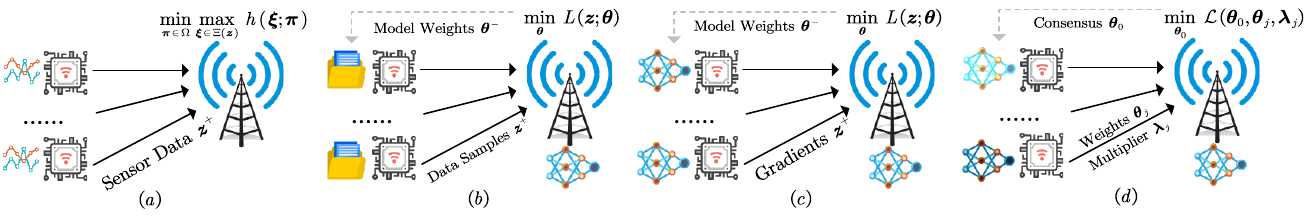}
	\caption{The communication paradigm for (a) data-driven decision making; (b) edge learning; (c) federated learning; (d) distributed optimization.}
	\label{fig:commu}
\end{figure*}
This subsection introduces two essential example applications on the RBA problem: data-driven decision-making and edge learning. They are either required to connect with large amounts of EDs or transmit substantial raw data samples.
\subsubsection{Data-driven Decision Making}
In data-driven decision-making, the CPS operator seeks to use real-time data as input parameters and boundary conditions to address its decision-making problem, which can be formulated as follows:
\begin{equation}\label{eq:decision}
\mathcal{C} \left( \boldsymbol{z} \right) =\min_{\boldsymbol{\pi }\in \Omega} \max_{\boldsymbol{\xi }\in \Xi (\boldsymbol{z})} \,\,h(\boldsymbol{\xi};\boldsymbol{\pi }).
\end{equation}
where $\boldsymbol{\pi}$ denotes the decisions of the CPS operator, constrained by the action set $\boldsymbol{\pi }\in \Omega$. The CPS operator decides its robust strategy decisions $\boldsymbol{\pi}$, according to the information given by the data $\boldsymbol{\xi }\in \Xi (\boldsymbol{z})$ at hand. Two simple examples are illustrated below to explain the meaning of Eq. \eqref{eq:decision}, one for smart grids and one for vehicle networks.

\textit{Example - Emergency demand response}:

In the context of smart grids, during incidents such as abrupt shutdown of power generators, the operator (often called the aggregator) in cities has to implement load shedding to prevent wider cascaded power grid failures. The operator must meet the load-shedding requirements while minimizing adverse economic costs. The objective function $h(\boldsymbol{\xi};\boldsymbol{\pi })$ is to minimize the cumulative economic cost associated with load shedding:
\begin{equation}\label{eq:edr1}
	\min_{\boldsymbol{\pi }\in \Omega} \max_{\boldsymbol{\xi }\in \Xi (\boldsymbol{z})} \,\,h(\boldsymbol{\xi};\boldsymbol{\pi })=\min_{\boldsymbol{\pi }\in \Omega}\max_{\boldsymbol{\xi }\in \Xi (\boldsymbol{z})} \,\,\sum\nolimits_j^{}{c_j\pi_j},
\end{equation}
where $\pi_j$ is the decision variable, indicating the load reduction amount (kW) of ED $j$; $c_j$ denotes the unit cost for instructing ED $j$ to execute a reduction. 

The action space for load reduction is articulated as follows:
\begin{equation}
	\Omega :\sum\nolimits_j^{}{\pi_j}\geqslant \pi_{\min},~~\pi_j\leqslant \xi _j,
\end{equation}
where the first constraint ensures that the total load reduction exceeds the minimum threshold $\pi_{\min}$; the second constraint confirms that load reduction for ED $j$ cannot be greater than its real-time load usage $\xi_j$.

The real-time load usage $\xi_j$ is  estimated either from the historical dataset or the real-time transmitting data $\boldsymbol{z}_j^+$:
\begin{equation}\label{eq:data_load}
	\Xi (\boldsymbol{z}): \begin{cases}
		\xi_j \stackrel{\text{d}}{=} \boldsymbol{z}_{\mathrm{old}},&s_j=0\\
		\xi_j=\boldsymbol{z}_j^+,&s_j=1
	\end{cases},
\end{equation}
where $\stackrel{\text{d}}{=}$ means the probability distribution on the two sides of equation is the same.

Without real-time data ($s_j=0$, the first scenario of Eq. \eqref{eq:data_load}), the operator resorts to estimating the available reduction load using historical data, assuming that $\xi_j$ follows the same distribution as $\boldsymbol{z}_{\mathrm{old}}$ (). Alternatively, if real-time data is accessible ($s_j=1$, the second scenario), the operator uses the measured real-time load data $z_j$ transmitted through ED $j$.\par
\textit{Example - Real-time vehicle routing}:\par
Under this setting, the CPS operator can be a delivery platform. It collects real-time traveling time data of each road segment within a city district and then plans routing paths for vehicles. The objective function $h(\boldsymbol{\xi};\boldsymbol{\pi })$ minimizes the total travel time to reach a destination:
\begin{equation}
\min_{\boldsymbol{\pi }\in \Omega} \max_{\boldsymbol{\xi }\in \Xi (\boldsymbol{z})} \,\,h(\boldsymbol{\xi};\boldsymbol{\pi })=\min_{\boldsymbol{\pi }\in \Omega} \max_{\boldsymbol{\xi }\in \Xi (\boldsymbol{z})} \sum_{mn}{\xi _{mn}}\pi_{mn},
\end{equation}
where $m$ and $n$ denote different road network nodes; $mn$ denotes the road interconnecting node $m$ and $n$; $\xi _{mn}$ is the estimated travel time for road segment $mn$; $\pi_{mn}$ is the binary decision variable indicating whether the road $mn$ is selected in the routing path. 

Traveling time, $\xi _{mn}$, is inferred  either from the historical dataset or the real-time transmitting data:
\begin{equation} \label{eq:data_vehicle}
\Xi (\boldsymbol{z}): \begin{cases}
	\xi _{mn}\stackrel{\text{d}}{=}\boldsymbol{z}_{\mathrm{old}},&s_j=0\\
	\xi _{mn}=z_{mn}^+,&s_j=1
\end{cases}.
\end{equation}
Lacking real-time data ($s_j=0$, the first scenario of Eq. \eqref{eq:data_vehicle}), the operator resorts to estimating the travel time using historical data. On the contrary, when real-time data is accessible ($s_j=1$, the second scenario), the operator directly employs the measured real-time travel time $z_{mn}$, transmitted by ED $j$. 

Vehicle routing conforms to road network flow balance constraints:
\begin{equation}
    \Omega :\begin{cases}
	\sum_n^{}{\pi_{mn}}-\sum_n^{}{\pi_{nm}}=b_m, &\forall m\in \mathcal{N}\\
	\pi_{mn}\in \{0,1\}, &\forall mn\in \mathcal{R}
\end{cases},
\end{equation}
where $b_m=1$ for source nodes, $b_m=-1$ for destination nodes, and $b_m=0$ for all other nodes.

As illustrated in Eqs. \eqref{eq:data_load} and \eqref{eq:data_vehicle}, the availability of real-time data $\boldsymbol{z}^+$ influences the actions of the CPS operator $\boldsymbol{\pi}$. Consequently, the communication system, which serves as the data pipeline, inherently affects the optimality of the decision. Data crucial to decisions should be prioritized. 
\begin{prop}\label{data}
    The total information utility gain for data-driven decision-making problems is equivalent to the decision cost reduction:
    \begin{equation}
    \begin{aligned}
    &\mathcal{C} \left( \boldsymbol{z}_{\mathrm{old}} \right) - \mathcal{C} \left( \boldsymbol{z}_{\mathrm{new}} \right) 
    \\
    =& \min_{\boldsymbol{\pi }\in \Omega} \max_{\boldsymbol{\xi }\in \Xi (\boldsymbol{z}_{\mathrm{old}})} \,\,h(\boldsymbol{\xi };\boldsymbol{\pi })-\min_{\boldsymbol{\pi }\in \Omega} \max_{\boldsymbol{\xi }\in \Xi (\boldsymbol{z}_{\mathrm{new}})} \,\,h(\boldsymbol{\xi };\boldsymbol{\pi }).
    \end{aligned}
    \end{equation}
\end{prop}
\begin{prop}\label{data_m}
    The marginal information utility gain $\Delta _j$ is the difference of decision costs between scenarios with and without accurate data from ED $j$:
    \begin{equation} \label{eq:margin1}
    \Delta _j=\min_{\boldsymbol{\pi }\in \Omega} \max_{\boldsymbol{\xi }\in \Xi (\boldsymbol{z}_{\mathrm{old}})} \,\,h(\boldsymbol{\xi };\boldsymbol{\pi })-\min_{\boldsymbol{\pi }\in \Omega} \max_{
    	\boldsymbol{\xi }_j=\boldsymbol{z}_j 
            \atop
    	\boldsymbol{\xi }_{j'}\in \Xi (\boldsymbol{z}_{\mathrm{old}}), j'\ne j} \,\,h(\boldsymbol{\xi };\boldsymbol{\pi }).
    \end{equation} 
\end{prop}

\subsubsection{Edge Learning}
In the context of edge learning, the operator deploys certain learning tasks on the edge server. The edge server collects raw data samples from different EDs and then accomplishes the learning tasks. Typical examples include training neural networks devised for recognizing road obstacles in vehicle networks and classifying cloud images for solar power prediction in smart grids. 

Suppose that ED $j$ has data samples $d=1,..,D_j$. The training data samples $\boldsymbol{z}_{j,d}$ can be further divided into $\boldsymbol{z}_{j,d}=\left( \boldsymbol{x}_{j,d},\boldsymbol{y}_{j,d}\right)$, where $\boldsymbol{x}_{j,d}$ is the input and $\boldsymbol{y}_{j,d}$ is the corresponding output. The learning model is expressed as $\hat{\boldsymbol{y}}_{j,d}=f\left( \boldsymbol{x}_{j,d};\boldsymbol{\theta } \right)$ where $\boldsymbol{\theta }$ denotes the model weights, and $\hat{\boldsymbol{y}}_{j,d}$ is the predicted output produced by the learned model. The operator aims to minimize the loss function for training the model as follows:  
\begin{equation} \label{eq:learn_obj}
\mathcal{C} \left( \boldsymbol{z} \right) =\min_{\boldsymbol{\theta }} \,L(\boldsymbol{z};\boldsymbol{\theta })=\min_{\boldsymbol{\theta }} \,\,\frac{1}{\sum_{j}^{}{D_j}}\sum_{j,d}^{}{{\ell \left( \boldsymbol{z}_{j,d};\boldsymbol{\theta } \right)}}.
\end{equation}

After collecting new data $\boldsymbol{z}_{\mathrm{new}}$, the CPS operator directs the edge server to update the model weights $\boldsymbol{\theta}$. Denote the model weights before and after updates as $\boldsymbol{\theta}^{-}$ and $\boldsymbol{\theta}^{+}$, respectively:
\begin{equation} \label{eq:model_weights}
	\boldsymbol{\theta }^-=\mathrm{arg} \min _{\boldsymbol{\theta }}\,\,L(\boldsymbol{z}_{\mathrm{old}};\boldsymbol{\theta }),~ \boldsymbol{\theta }^+=\mathrm{arg} \min _{\boldsymbol{\theta }}\,\,L(\boldsymbol{z}_{\mathrm{new}};\boldsymbol{\theta }). 
\end{equation}

The goal of the CPS operator is to minimize the loss $L$ across all data samples. Given that test data is expected to mirror the distribution of the training data set, a lower loss function value typically indicates higher model accuracy.
\begin{prop}\label{edge}
	The information utility gain for edge learning equals to the reduction in the learning model's loss function for all data samples:
	\begin{equation}\label{eq:prop2}
		\begin{aligned}
			&\mathcal{C} \left( \boldsymbol{z}_{\mathrm{old}} \right) -\mathcal{C} \left( \boldsymbol{z}_{\mathrm{new}} \right) \\
			&=\frac{1}{\sum_j^{}{D_j}}\left(\sum_{j,d}^{}{{\ell \left( \boldsymbol{z}_{j,d};\boldsymbol{\theta }^- \right)}}-\sum_{j,d}^{}{{\ell \left( \boldsymbol{z}_{j,d};\boldsymbol{\theta }^+ \right)}}\right),
		\end{aligned}
	\end{equation}
	where the model weights $\boldsymbol{\theta}^{-}$ and $\boldsymbol{\theta}^{+}$ are determined by the collected data via Eqs. \eqref{eq:model_weights}.
\end{prop}

To precisely compute the information utility gain Eq. \eqref{eq:prop2}, the CPS operator must re-train the learning model to update the weights $\boldsymbol{\theta}^+$. Unfortunately, the CPS operator cannot precisely compute Eq. \eqref{eq:prop2} before it collects new data $\boldsymbol{z}_{\mathrm{new}}$. To solve the issue, the following approximation method is proposed:
\begin{align}
	&\mathcal{C} \left( \boldsymbol{z}_{\mathrm{old}} \right) -\mathcal{C} \left( \boldsymbol{z}_{\mathrm{new}} \right)
	\\
	 \approx&\frac{1}{\sum_j^{}{D_j}}\left( \textstyle\sum_{\{ \left. j,d :\boldsymbol{z}_{j,d}\notin \boldsymbol{z}_{\mathrm{old}} \right\}}^{}{\ell \left( \boldsymbol{z}_{j,d};\boldsymbol{\theta }^- \right)}-\right. \nonumber
	 \\
	 &~~~~~~~~~~\left.\textstyle\sum_{\left\{  j,d :\boldsymbol{z}_{j,d}\notin \boldsymbol{z}_{\mathrm{new}} \right\}}^{}{\ell \left( \boldsymbol{z}_{j,d};\boldsymbol{\theta }^+ \right)} \right) 
	\\
	\approx&\frac{1}{\sum_j^{}{D_j}}\left( \sum\textstyle_{\left\{ j,d:\boldsymbol{z}_{j,d}\in \boldsymbol{z}_{\mathrm{new}}\cap \boldsymbol{z}_{j,d}\notin \boldsymbol{z}_{\mathrm{old}} \right\}}^{}{\ell \left( \boldsymbol{z}_{j,d};\boldsymbol{\theta }^- \right)} \right).
\end{align}

The first approximation is based on the fact that, if the CPS operator knows the exact value of $\boldsymbol{z}_{j,d}$ and includes it in the training set, the model's loss $\ell \left( \boldsymbol{z}_{j,d};\boldsymbol{\theta } \right)$ for query data $\boldsymbol{z}_{j,d}$ after training will be small and thus can be approximated as zero.  The major loss arises from unknown data patterns, namely data $\left\{ j,d :\boldsymbol{z}_{j,d}\notin \boldsymbol{z}_{\mathrm{old}} \right\} $ for the old model $\boldsymbol{\theta}^-$ and data $\left\{  j,d:\boldsymbol{z}_{j,d}\notin \boldsymbol{z}_{\mathrm{new}} \right\} $ for the updated model $\boldsymbol{\theta}^+$.  

The second approximation follows that the major gains come from the data patterns present in the new dataset but not in the old one ($\left\{ j,d :\boldsymbol{z}_{j,d}\in \boldsymbol{z}_{\mathrm{new}}\cap \boldsymbol{z}_{j,d}\notin \boldsymbol{z}_{\mathrm{old}} \right\} $) seen by the old model $\boldsymbol{\theta}^-$. For such data,  the loss $\ell \left( \boldsymbol{z}_{j,d};\boldsymbol{\theta }^+ \right)$ is approximately 0 while $\ell \left( \boldsymbol{z}_{j,d};\boldsymbol{\theta }^- \right)$ remains nonzero. Consequently, the marginal information utility gain for each ED can be approximated as follows: 
\begin{prop}\label{edge_m}
	The marginal information utility gain for edge learning is the loss of the existing learning model with respect to new data samples:
	\begin{equation}\label{eq:margin2}
		\Delta _{j,d}=\ell \left( \boldsymbol{z}_{j,d};\boldsymbol{\theta }^- \right).
	\end{equation} 
\end{prop}
This method is consistent with the principles of \textit{active learning}. Collecting data that the existing model ($\boldsymbol{\theta}^-$) is least confident about, especially those ambiguous and near the decision boundary, allows training with fewer samples. For example, for hinge loss (usually used in support vector machines), data samples close to the support plane induce higher loss, therefore a higher priority to be inquired. For multiclass classifications, data samples with higher predicted entropy are crucial for enhancing accuracy and are therefore given higher inquiry priority.

When the model weights are small and the EDs possess computing capabilities, the model can be broadcast via the downlink channel from the edge server to each ED. Following each training round, the edge server broadcasts the model weights $\boldsymbol{\theta}^-$ to all EDs, just as shown in Fig. \ref{fig:commu}. Upon receiving the model weights, EDs can compute Eq. \eqref{eq:margin2}. When dealing with large model weights, the CPS operator uses the expected information utility gains from each ED as a prediction of future gains, similar to the approach in Problem (P3).

\subsection{Goal-oriented RBA for Intermediate Data Transmission}
This section presents two additional typical examples of Problem (P1): federated learning and distributed optimization. For privacy reasons, the data transmitted during communications can only include intermediate results (such as the gradient). Thus, it mandates the adoption of a distributed approach, entailing numerous iterations and rounds of communications for convergence. One should distinguish between two forms of data involved in these processes: `raw data' and `transmitted data'. The raw data, not intended for transmission, are denoted by $\mathsf{z}$. The transmitted data are denoted as previously by $\boldsymbol{z}$.

\subsubsection{Federated Learning}
In the context of federated learning, various EDs collaborate to jointly train a learning model, adhering to a similar loss function as that in edge learning:
\begin{equation}
\min_{\boldsymbol{\theta }} \,L(\mathsf{z};\boldsymbol{\theta })=\min_{\boldsymbol{\theta }} \,\,\frac{1}{\sum_j^{}{D_j}}\sum_{j,d}^{}{{\ell \left( \mathsf{z}_{j,d};\boldsymbol{\theta } \right)}},
\end{equation}
where $\mathsf{z}_{j,d}$ denotes the raw data samples at ED $j$ for training. To protect privacy, EDs compute the gradients of the loss function with respect to the model weights locally and transmit gradients rather than raw data samples to the server. The server aggregates these gradients, updates the model weights, and transmits them to the EDs. This procedure is repeated for $K$ rounds ($k=1,...,K$)  to facilitate the convergence of the final weights.

During round $k$, ED $j$ computes the gradient $\boldsymbol{g}_{j,k}$ and sends the gradient data $\boldsymbol{z}_{j,k}^+$ as follows:  
\begin{equation}
\boldsymbol{g}_{j,k}=\boldsymbol{z}_{j,k}^+=\frac{1}{D_j}\sum_d^{}{\nabla _{\boldsymbol{\theta }}\ell\left( \mathsf{z}_{j,d};\boldsymbol{\theta }_{k-1} \right)}.
\end{equation}

The server collects all the gradient updates, then computes the aggregated gradient as $\boldsymbol{g}_k=\frac{\sum\nolimits_j^{}{D_j\boldsymbol{g}_{j,k}}}{\sum\nolimits_j^{}{D_j}}$. Then the operator employs gradient descent to update the model weights $\boldsymbol{\theta }_k$ as:
\begin{equation}\label{eq:ga}
\boldsymbol{\theta }_k=\boldsymbol{\theta }_{k-1}-\eta _k\boldsymbol{g}_k=\boldsymbol{\theta }_{k-1}-\eta _k\frac{\sum\nolimits_j^{}{D_j\boldsymbol{g}_{j,k}}}{\sum\nolimits_j^{}{D_j}},
\end{equation}
where $\eta _k$ is the learning rate at round $k$. 

Typically, it is assumed that the loss function satisfies certain desirable conditions. The reduction in the loss function can be bounded as shown below:
\begin{lemma}\label{lemma1}
If the loss functions $L_{\boldsymbol{\theta }}(\mathsf{z})$ is $\eta$-smooth, that is, for $\forall \boldsymbol{\vartheta }, \boldsymbol{\theta}$,
\begin{equation}
L(\mathsf{z};\boldsymbol{\vartheta })-L(\mathsf{z};\boldsymbol{\theta })\le \left< \nabla _{\boldsymbol{\theta }}L,\boldsymbol{\vartheta }-\boldsymbol{\theta } \right> +\frac{\kappa}{2}\| \boldsymbol{\vartheta }-\boldsymbol{\theta }\| _{2}^{2}.
\end{equation}
The loss function reduction between any two consecutive rounds can be bounded by:
\begin{equation}
	L(\mathsf{z};\boldsymbol{\theta }_k)-L(\mathsf{z};\boldsymbol{\theta }_{k-1})\leqslant -\eta _k\left( 1-\frac{\kappa \eta _k}{2} \right) \left\| \boldsymbol{g}_k \right\| _{2}^{2}.
\end{equation}	
when $0<\eta_k \leq \frac{2}{\kappa}$, the loss function always descends. 
\end{lemma}
The $\kappa$-smooth condition indicates that the loss function gradient will not change drastically when $\boldsymbol{\theta }$ changes. This condition holds for many types of loss function and is often used in convergence analysis.

If only part of EDs $j\in \{ j:s_j=1 \}$ update their gradients, the aggregated weights will be updated as:
\begin{equation}\label{eq:partial_ga}
\boldsymbol{\theta }_k=\boldsymbol{\theta }_{k-1}-\eta _k\tilde{\boldsymbol{g}}_k=\boldsymbol{\theta }_{k-1}-\eta _k\frac{\sum\nolimits_{j\in \{ j:s_j=1 \} }^{}{D_j\boldsymbol{g}_{j,k}}}{\sum\nolimits_{j\in \{ j:s_j=1 \} }^{}{D_j}}.
\end{equation}
\begin{lemma}
	For the partial update, the reduction of the loss function between two consecutive rounds can be bounded by:
	\begin{equation}
		\begin{aligned}
		&L(\mathsf{z};\boldsymbol{\theta }_k)-L(\mathsf{z};\boldsymbol{\theta }_{k-1})\leqslant 
		\\
		&-\eta _k\left( 1-\frac{\kappa \eta _k}{2} \right) \left\| \tilde{\boldsymbol{g}}_k \right\| _{2}^{2}-\eta _k\left< \boldsymbol{g}_k-\tilde{\boldsymbol{g}}_k,\tilde{\boldsymbol{g}}_k \right> .
		\end{aligned}
	\end{equation}
	when $\eta_k \leq \frac{2}{\kappa}$ and $\left< \boldsymbol{g}_k-\tilde{\boldsymbol{g}}_k,\tilde{\boldsymbol{g}}_k \right>  \geqslant 0$, the loss function is always descending. 
\end{lemma}
Generally, $\boldsymbol{g}_k$ and $\tilde{\boldsymbol{g}}_k$ are similar because the data samples $\mathsf{z}_{j}$ of different devices follow similar distributions. Thus, $\left< \boldsymbol{g}_k-\tilde{\boldsymbol{g}}_k,\tilde{\boldsymbol{g}}_k \right>  \approx 0$. Sometimes, it is assumed that the expected values of the two gradients are the same, meaning that the expected value of the last term is zero.
\begin{prop} \label{federated}
	The information utility gain for federated learning in iteration round $k$ is approximately equivalent to the aggregated gradient norm:
	\begin{equation}
		 \mathcal{C} \left( \boldsymbol{z}_{\mathrm{old}} \right) - \mathcal{C} \left( \boldsymbol{z}_{\mathrm{new}} \right)  \approx  \eta _k\left( 1-\frac{\kappa \eta _k}{2} \right) \left\| \tilde{\boldsymbol{g}}_{k} \right\| _{2}^{2}.
	\end{equation}
\end{prop}
The expression above is not separable and may not satisfy the submodular property. It is generally assumed that data samples from different EDs follow similar distributions. As a result, the inner products of the gradients of any two EDs tend to align in a similar direction. Thereby, $\left< \boldsymbol{g}_{j_1,k},\boldsymbol{g}_{j_2,k} \right> \geqslant 0, \forall j_1\ne j_2$ is generally true, and we have:
\begin{equation}
	\left\| \tilde{\boldsymbol{g}}_k \right\| _{2}^{2}\geqslant \sum_{j\in \{ j:s_j=1 \} }^{}{\left\| \frac{D_j\boldsymbol{g}_{j,k}}{\sum\nolimits_{j\in \{ j:s_j=1 \} }^{}{D_j}} \right\| _{2}^{2}}.
\end{equation}

Thereby, the marginal information utility gain for federated learning can be expressed as:
\begin{prop} \label{federated_m}
   The marginal information utility gain for federated learning at iteration round $k$ is approximately equivalent to the weighted gradient norm:
    \begin{equation}\label{eq:margin3}
    	\Delta _j=\eta _k\left( 1-\frac{\kappa \eta _k}{2} \right) \left\| \frac{D_j\boldsymbol{g}_{j,k}}{\sum\nolimits_{j\in \{ j:s_j=1 \} }^{}{D_j}} \right\|_2 ^2.
    \end{equation}
\end{prop}

\subsubsection{Distributed Optimization}
\begin{algorithm}[t] 
	\caption{Standard Distributed ADMM}\label{al:standard_algorithm}
	\SetAlgoLined
	\KwInput {Initial values $\boldsymbol{\lambda}_{j\left( 0 \right)},\boldsymbol{\theta}_{0\left( 0 \right)}, \boldsymbol{\theta}_{j\left( 0 \right)}$; stopping criterion $\varepsilon_1,\varepsilon_2$; $k=1$}
	\While {$\| \Delta\boldsymbol{\lambda }_{j\left( k,k-1 \right)}\| \geqslant \varepsilon_1$ \textbf{or} $ \| \Delta\boldsymbol{\theta}_{j\left( k,k-1 \right)}\| \geqslant \varepsilon_2$ }{
		\textbf{Step~1:} The operator server solves 
		\begin{equation} \label{eq:step1}
		\boldsymbol{\theta}_{0,k+1} = \mathrm{argmin} _{\boldsymbol{\theta }_0}\,\,\mathcal{L} ( \boldsymbol{\theta }_0,\boldsymbol{\theta}_{j,k},\boldsymbol{\lambda}_{j,k}) 		    
		\end{equation}
		and transmits variables $\boldsymbol{\theta}_{0(k+1)}$ to EDs \;
		\textbf{Step~2:} EDs receive $\boldsymbol{\theta}_{0(k+1)}$ and update $\boldsymbol{\theta}_{j(k+1)}$
  		\begin{equation}\label{eq:step2}
		\boldsymbol{\theta}_{j(k+1)} = \mathrm{argmin} _{\boldsymbol{\theta }_j}\,\,\mathcal{L} ( \boldsymbol{\theta }_{0,k+1},\boldsymbol{\theta}_{j},\boldsymbol{\lambda}_{j,k})     
		\end{equation}\\
		\textbf{Step~3:} EDs update multipliers $\boldsymbol{\lambda }_{j,k+1}$ as:
    	\begin{equation}\label{eq:step3}
		\boldsymbol{\lambda }_{j,k+1}=\boldsymbol{\lambda }_{j,k}+\rho \left( \boldsymbol{\theta}_{j,k+1}-\boldsymbol{\theta}_{0,k+1} \right) 
		\end{equation}
		And send updates $\boldsymbol{\theta}_{j,k+1}, \boldsymbol{\lambda }_{j,k+1}$ to the VPP.\\
		\textbf{Step~4:} The server waits until it receives all EDs' updates and sets $k=k+1$.   
	}
	\KwResult {$\boldsymbol{\theta}_{0,k}, \boldsymbol{\theta}_{j,k},\boldsymbol{\lambda}_{j,k}$}
\end{algorithm}
Distributed optimization frequently arises in domains of CPSs. For example, in vehicle networks, fleets can communicate to reach a consensus on obstacles and location. Similarly, in smart grids, energy users in a district might communicate to negotiate energy trading. Among various methods for distributed solutions, primal-dual decomposition stands out, with the Alternating Direction Method of Multipliers (ADMM) often used due to its robust convergence guarantees \cite{boyd2011distributed}.

With a little abuse of notation, suppose now that every ED has its own copy of decision variables $\boldsymbol{\theta }_j $ and the server has an individual copy of decision variable $\boldsymbol{\theta }_0$. The operator and EDs try to reach a consensus on the decision variables while minimizing the summed cost :
\begin{equation} \label{eq:consensus}
\min_{\boldsymbol{\theta }_j =\boldsymbol{\theta }_0} \,\sum\nolimits_j{L_j(\mathsf{z}_j;\boldsymbol{\theta }_j)},
\end{equation}
where $L_j(\mathsf{z}_j;\boldsymbol{\theta }_j)$ is the cost function of ED $j$. This optimization problem now introduces a consensus constraint $\boldsymbol{\theta }_j=\boldsymbol{\theta }_0$. The augmented Lagrangian function can then be written as:
\begin{equation} \label{eq:lag}
	\begin{aligned}
	\,\,  &\mathcal{L} \left( \boldsymbol{\theta }_0,\boldsymbol{\theta }_j,\boldsymbol{\lambda }_j \right) 
	\\
	=&\sum\nolimits_j^{}{L_j(\mathsf{z}_j;\boldsymbol{\theta }_j)+\left< \boldsymbol{\lambda }_j,\boldsymbol{\theta }_j-\boldsymbol{\theta }_0 \right> +\frac{\rho}{2}\left\| \boldsymbol{\theta }_j-\boldsymbol{\theta }_0 \right\| _{2}^{2}},
	\end{aligned}
\end{equation}
where $\boldsymbol{\lambda }_{j}$ denote the dual multipliers and $\rho$ is  the penalty factors. Using the decomposition from Eq. \eqref{eq:lag}, ADMM addresses the problem as shown in Algorithm \ref{al:standard_algorithm}, where the transmitted data include local primal decision variables $\boldsymbol{\theta}_j$ and dual multipliers $\boldsymbol{\lambda}_j$.

When only parts of the EDs $j\in \{ j:s_j=1 \}$ participate in the update, only $j\in \{ j:s_j=1 \}$ execute steps 2-3. Other EDs $j\in \{ j:s_j=0 \}$ remain inactive and their primary and dual variables remain unchanged $\boldsymbol{\theta }_{j,k+1}=\boldsymbol{\theta }_{j,k},\boldsymbol{\lambda }_{j,k+1}=\boldsymbol{\lambda }_{j,k}$. In the partial update scheme, given a convex and smooth objective function and a properly selected penalty factor, we can prove that the difference in the Lagrangian function between two consecutive rounds meets the following conditions:
\begin{prop} \label{distributed}
    The information utility gain for distributed ADMM optimization at iteration round $k$ can be approximated via the sum of the variable changes and can be expressed as:
    \begin{equation}
    \begin{aligned}
        \,\, &\mathcal{C} \left( \boldsymbol{z}_{\mathrm{old}} \right) - \mathcal{C} \left( \boldsymbol{z}_{\mathrm{new}} \right) 
        \\
        &=\mathcal{L} \left( \boldsymbol{\theta }_{0,k},\boldsymbol{\theta }_{j,k},\boldsymbol{\lambda }_{j,k} \right) - \mathcal{L} \left( \boldsymbol{\theta }_{0,k+1},\boldsymbol{\theta }_{j,k+1},\boldsymbol{\lambda }_{j,k+1} \right) 
        \\
        &\ge \sum_{j\in \{ j:s_j=1 \}}{\alpha}\left\| \boldsymbol{\theta }_{j,k+1}-\boldsymbol{\theta }_{j,k} \right\|_2 ^2+\beta \left\| \boldsymbol{\theta }_{0,k+1}-\boldsymbol{\theta }_{0,k} \right\|_2 ^2,
    \end{aligned}
    \end{equation}
    where $\alpha>0$ and $\beta>0$ are two positive constants.
\end{prop}
The proof can be found in the Appendix. When the Lagrangian function $\mathcal{L} \left( \boldsymbol{\theta }_0,\boldsymbol{\theta }_j,\boldsymbol{\lambda }_j \right) $ has a lower bound and follows a monotonically decreasing sequence, the variables $\boldsymbol{\lambda}_{j,k},\boldsymbol{\theta}_{0,k}, \boldsymbol{\theta}_{j,k}$ converge to the Karush-Kuhn-Tucker (KKT) point of the problem Eq. \eqref{eq:consensus}. Consequently, by selecting ED data aligned with a larger Lagrangian reduction, we inherently opt for the steepest descent toward the KKT point.

\begin{prop}  \label{distributed_m}
    The marginal information utility gain for distributed ADMM optimization at iteration round $k$ can be approximated by the sum of variable changes for each ED $j$:
    \begin{equation}\label{eq:margin4}
    	\Delta _{j}=\alpha \left\| \boldsymbol{\theta }_{j,k}-\boldsymbol{\theta }_{j\left( k-1 \right)} \right\|_2 ^2.
    \end{equation} 
\end{prop}
The second term is neglected as an approximation. The approximation is reasonable. Large changes in individual decision variables, $ \| \boldsymbol{\theta }_{j,k}-\boldsymbol{\theta }_{j\left( k-1 \right)} \|^2_2$, will also lead to significant changes in consensus, as $\boldsymbol{\theta }_{0,k}$ adjusts in response to the new consensus solution.

\section{Simulation Results} \label{case}
In this section, we evaluate the effectiveness of our proposed RBA framework. The simulation parameters are set as follows unless otherwise specified. According to LTE standards, each RB is $B=180$ kHz wide in frequency and 1 slot ($t=0.5$ ms) long in time. We assume that there are a total of $3$ MHz or $15$ RBs per slot that is delicately allocated to the CPS operator, which it can freely allocate. The scheduling time interval is set to be 1 second and the channel scheduler will re-schedule the RB allocation every second. The channel noise spectral density is set to $-52.7$ dBW / Hz over the band for each ED. The channel noise is set to be $B\sigma^2=1$ over the band for each ED. All channels $g$ are assumed to be i.i.d. Rayleigh fading with scale parameter 1. By default, we assume that the accurate channel gain can be observed and that the channel gain will vary every second. The unscaled transmitting power of EDs is set as $p_j=1$ W \cite{Han2023}.
\subsection{Application Setting}
\subsubsection{Data-driven Decision Making}
We set up the case with a large number of accesses and with small-size data from EDs. Algorithm \ref{al:mmtc} will be used to allocate RBs. Referring to the emergency demand response problem detailed in Eqs. \eqref{eq:edr1}-\eqref{eq:data_load}, We have a total of $J=15000$ EDs capable of reducing energy demand. The cost for each ED is set at random, in the range $c_j\in[0,5]$ \$/kW. The overall minimum energy reduction is $\pi_{\min}=10^4$ kW. The reduction amount that each ED can offer varies, but it's uniformly distributed between $\xi _j\in[1,\xi_{j,\max}]$ kW. The maximum reduction $\xi_{j,\max}$ is also random, but it won't exceed $30$ kW. Past data or the old dataset follows this same distribution.  Even though the main data might be small, there's often extra or redundant data in the packet head. So, we consider the size of one sensor data packet to be $64$ bytes.

\subsubsection{Edge Learning}
We set up the case with large-size data from EDs. Algorithm \ref{al:embb} will be used to allocate RBs. The learning task is set to classify images, and we use the widely-recognized MNIST dataset. It consists of $10$ categories ranging from digit `0' to `9' and a total of $60000$ labeled training data samples. Each image is 28 by 28 pixels, resulting in a size of 784 pixels. We divide the dataset as $50000$ samples of the training set and $10000$ samples of the test set. We consider $J = 10$ EDs and each ED has $10000$ data samples. We consider a non-i.i.d data distribution case: the data samples that contain the digits `6' and `9' are mainly included in two EDs. We train a multilayer perceptron which has a $784$-unit input layer with ReLU activation, a $64$-unit hidden layer, and a $10$-unit softmax output layer, with $50890$ weights in total~\cite{Sun2022}. The initial learning rate $\eta$ is set to $0.01$. A momentum of $0.9$ is adopted. The cross-entropy is adopted as the loss function. The mini-batch size is $512$, and the number of epochs is $20$. 

\subsubsection{Federated Learning}
In the context of federated learning, we maintain consistency with the learning objectives and parameters previously mentioned for edge learning. Algorithm \ref{al:embb} will be used to allocate RBs. For each update round within this federated learning framework, every ED will execute a single iteration of stochastic mini-batch gradient descent.
\subsubsection{Distributed Optimization}
We set up the case with large-size data from EDs. Algorithm \ref{al:embb} will be used to allocate RBs. For distributed optimization, we consider $J = 10$ EDs working collaboratively on the task of system identification. System identifications are commonly used in areas such as wireless location detection and smart grid controls. For system identification, the EDs collectively collect data according to the law of $\mathbf{Y}_j=\boldsymbol{\theta}\mathbf{ X}_j+\mathbf{N}_j$. $\mathbf{Y}_j$ and $\mathbf{X}_j$ represent the observation data and state data acquired by ED $j$, respectively. Meanwhile, $\mathbf{N}_j$ denotes the observation noise associated with ED $j$, which is modeled by a Gaussian distribution. Each ED aims to estimate the matrix $\boldsymbol{\theta}$ based on its individual observations. The objective function for each ED is:
\begin{equation}
	L_j(\mathsf{z}_j;\boldsymbol{\theta }_j)=\frac{1}{2}\| \mathbf{Y}_j-\boldsymbol{\theta }_j\mathbf{X}_j\| _{F}^{2}+\varrho \| \boldsymbol{\theta }_j\| _1.
\end{equation}

$\boldsymbol{\theta}$ is set as a $100\times 100$ random matrix with $50\%$ elements being zeros. Each ED has a total of $30$ data samples of $\mathbf{Y}_j$ and $\mathbf{X}_j$. The noise is assumed to follow a normal distribution with zero mean and variance of $0.015\times j$. This configuration ensures that different EDs obtain observations of varying quality. The sparsity regulation penalty is set as $\varrho =0.1$, and the ADMM penalty is set as $\rho=0.1$.
\subsubsection{Benchmark Setting}
We utilize three RBA policies:
\begin{itemize}[leftmargin=8pt]
	\item Channel-based RBA: RBs are assigned based on the quality of the channel. The ED with the best channel quality receives sufficient RB allocations to transmit its data, followed by the next best, and so on. This method aims to maximize network throughput and is consistent with current cellular network practices. 
	\item Utility-based RBA: Channels are assigned based on the ED data utility gain. The ED with the highest marginal information utility gain is provided enough RB allocations for data transmission, then the next highest, and so on. This method does not consider poor channel quality. 
	\item Hybrid RBA: This method uses both Algorithm \ref{al:embb} and Algorithm \ref{al:mmtc}. Termed 'hybrid', it factors in both the channel quality and information utility when allocating RBs.
\end{itemize}

\subsection{Data-driven Decision Making}
\begin{figure}[t]
	\centering
	\includegraphics[width=0.49\textwidth]{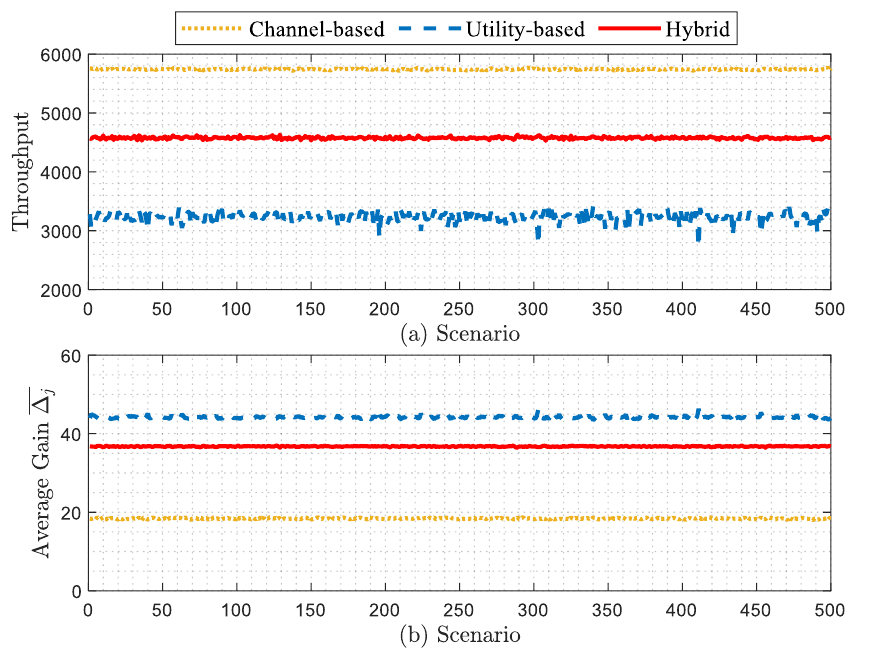}
	\caption{(a) The throughput and (b) the relative information utility gain of communication networks under different policies in different scenarios (data-driven decision-making).}
	\label{fig:1-1}
\end{figure}
\begin{figure}[t]
	\centering
	\includegraphics[width=0.49\textwidth]{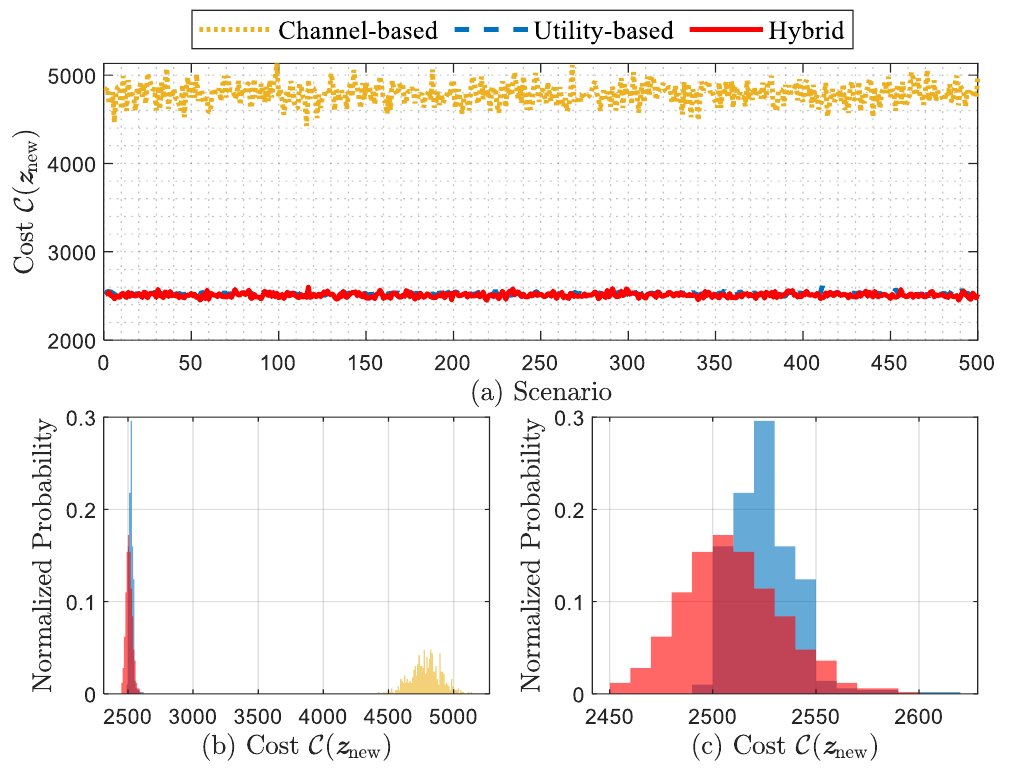}
	\caption{(a) The decision cost under different policies in different scenarios; (b) the probability distribution of decision cost for the $500$ scenarios for three policies; (c) the probability distribution of decision cost for the $500$ scenarios for utility-based and hybrid based policies (data-driven decision making).}
	\label{fig:1-2}
\end{figure}
We simulate the ED data collection in $500$ scenarios, each presenting a different channel gain for each ED. For the 500 scenarios, the throughput and the relative information utility gain are shown in Fig. \ref{fig:1-1}(a) and (b), respectively. The throughput is measured as the number of EDs that are scheduled to upload their data. Notably, the channel-based policy achieves the maximum throughput while delivering the minimum mean information utility gain. In contrast, the utility-based policy achieves the maximum mean information utility gain while throughput performance is poor. The hybrid policy delivers mediate performance between the two.

The goal of the CPS operator is to reduce decision costs. In each scenario, after acquiring the new data $\boldsymbol{z}_{\mathrm{new}}$, we solve the optimization problem described in Eqs. \eqref{eq:edr1}-\eqref{eq:data_load}. The cost $\mathcal{C}(\boldsymbol{z}_{\mathrm{new}})$ under different scenarios are depicted in Fig. \ref{fig:1-2}(a-c). From Fig. \ref{fig:1-1}(a-b), it is noteworthy that the channel-based policy fails to guarantee a low decision cost. The system costs under both the utility-based policy and the hybrid policy are close to each other. Fig. \ref{fig:1-2}(c) shows that the hybrid policy can guarantee a lower system cost than the utility-based policy. The hybrid policy achieves a balance between data volume and data utility.  Compared with the channel-based policy, the hybrid policy can reduce approximately $50\%$ of possible decision costs under the same communication budget.

\subsection{Edge Learning}
\begin{figure}[t]
	\centering
	\includegraphics[width=0.49\textwidth]{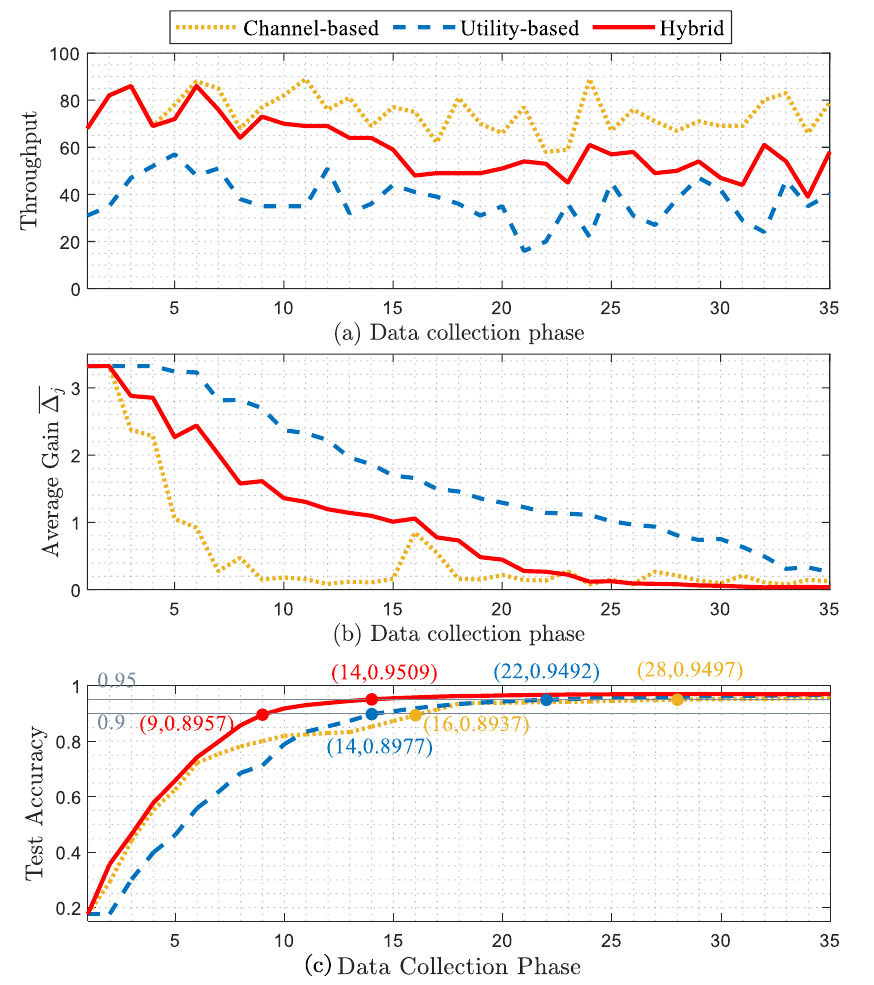}
	\caption{(a) The throughput and (b) the relative information utility gain of communication networks under different policies in different data collection rounds (edge learning). (c) The test accuracy under different data collection rounds. Marked points denote the points that are closest to the $0.9$ and $0.95$ accuracy level line. The numbers in the bracket denote (data collection round, test accuracy).}
	\label{fig:2-1}
\end{figure}
We conducted edge learning simulations across $35$ phases of data collection. Initially, the server did not have data samples. For each phase, it collects new data and retrains the learning model. For the 35 phases of data communication, the throughput and the relative information utility gain are shown in Fig. \ref{fig:2-1}(a) and Fig. \ref{fig:2-1}(b), respectively. The throughput is measured as the number of raw data samples/images that are uploaded to the edge server. The policies behave in a similar way to those in the data-driven decision-making case. 

The goal of the CPS operator is to minimize loss function and enhance model accuracy. In each data collection phase, we test the inference accuracy of the model on the test set, and the result is shown in Fig. \ref{fig:2-1}(c). In Fig. \ref{fig:2-1}(c), we also marked when the test accuracy under three different policies is closest to the accuracy $0.9$ and $0.95$. In the initial phase, the advantages of abundant data might surpass those of valuable data. The test accuracy under the hybrid policy and the channel-based policy are close to each other. However, as the dataset accumulates, the benefit of useful data may suppress the benefit of more data. Hence, the utility-based policy's test accuracy surpasses the channel-based ones. Nevertheless, the test accuracy of the hybrid policy consistently leads. Compared to channel-based RBA, the hybrid policy saves approximately $50\%$ iteration rounds (total training time) to achieve the precision of the $0.9$ and $0.95$ tests.

\subsection{Federated Learning}
\begin{figure}[t]
	\centering
	\includegraphics[width=0.49\textwidth]{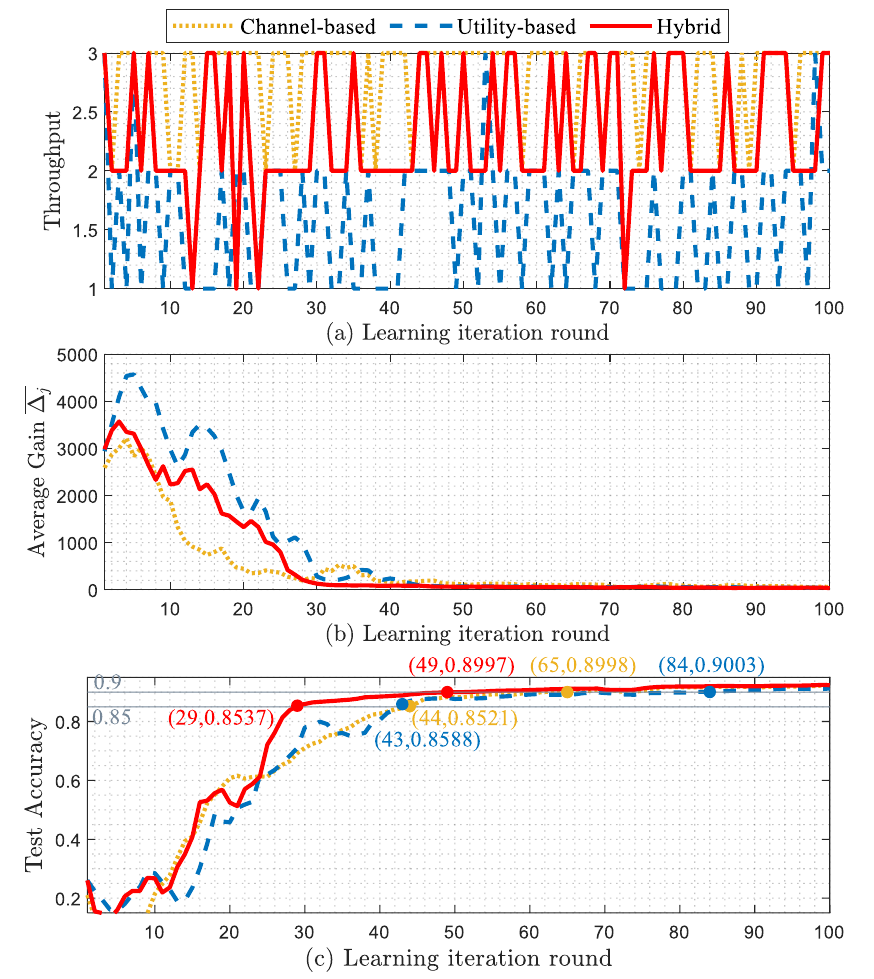}
	\caption{(a) The throughput and (b) the relative information utility gain of communication networks under different policies in different learning iteration rounds (federated learning). (c) The test accuracy under different learning iteration rounds. Marked points denote the points that are closest to the $0.9$ and $0.95$ accuracy level line. The numbers in the bracket denote (iteration round, test accuracy).}
	\label{fig:3-1}
\end{figure}
We simulate federated learning over a total of 100 iteration rounds. The server begins with an initial weight. In each round, the operator aggregates the ED gradients and performs gradient descent updates. For the 100 iteration rounds, the throughput and the relative information utility gain are shown in Fig. \ref{fig:3-1}(a) and Fig. \ref{fig:3-1}(b), respectively. The throughput is measured as the number of EDs that uploaded their gradients to the edge server. The policies demonstrate trends analogous to those in the data-driven decision-making case. 

The goal of the CPS operator is to reduce the loss function and enhance model accuracy. In each learning round, we test the inference accuracy of the model on the test set and record the loss function value, shown in Fig. \ref{fig:3-1}(c). In Fig. \ref{fig:3-1}(c), we also highlight when the test accuracy under three different policies approaches accuracy $0.9$ and $0.95$. Federated learning performs similarly to edge learning. In the initial phase, the benefit of more gradient information may provide more accurate gradient descent information. The test accuracy under the hybrid policy and the channel-based policy are close to each other. However, as the dataset accumulates, the benefit of high-norm gradient information may suppress the benefit of more gradient information. This might stem from the model being already around its convergence region. The crucial aspect is identifying the steepest direction towards the minimum.  Compared with the channel-based policy, the hybrid policy can reduce approximately $30\%$ of the total training time to reach $0.9$ and $0.95$ test accuracy under the same communication budget. 

\subsection{Distributed Optimization}
\begin{figure}[t]
	\centering
	\includegraphics[width=0.49\textwidth]{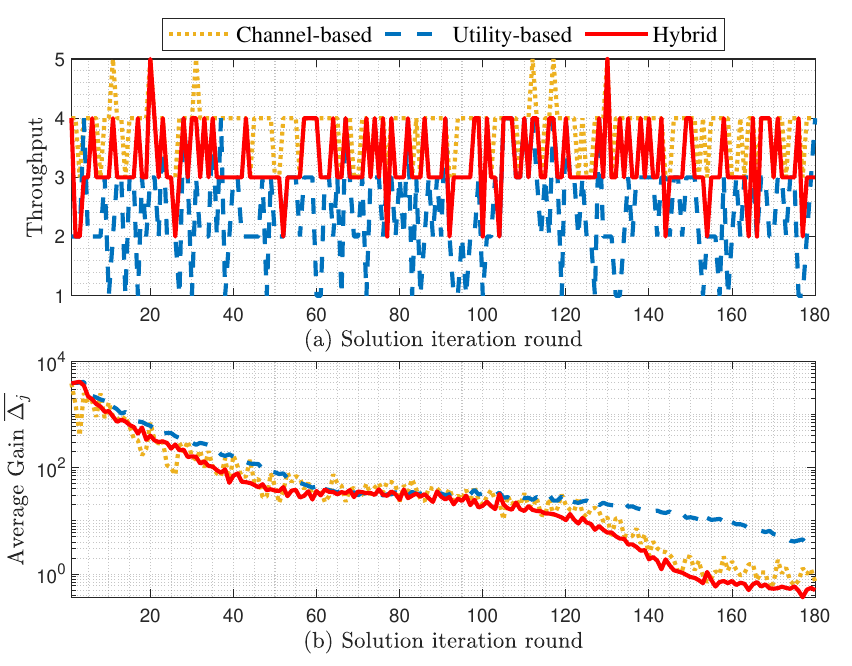}
	\caption{(a) The throughput and (b) the relative information utility gain of communication networks under different policies in different solution iteration rounds (distributed optimization).}
	\label{fig:4-1}
\end{figure}
\begin{figure}[t]
	\centering
	\includegraphics[width=0.49\textwidth]{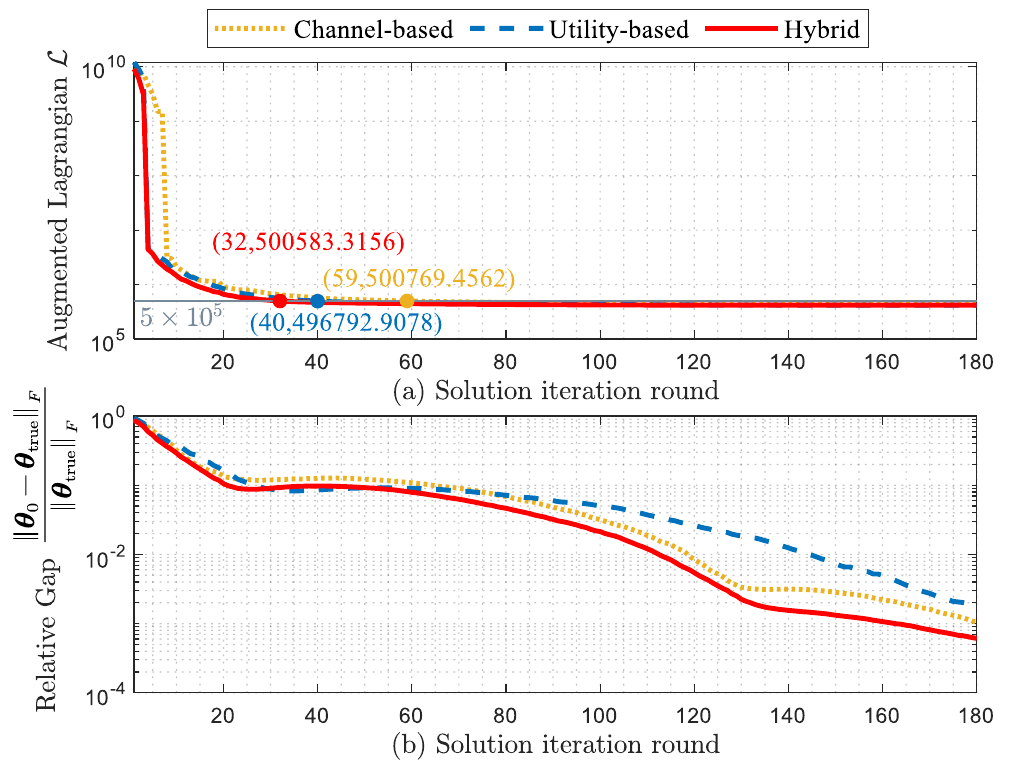}
	\caption{(a) The value of the augmented Lagrangian function and (b) the relative gap under different learning solution rounds. Marked points denote the points that are closest to the $5\times 10^5$ level line. The numbers in brackets denote (iteration round, Lagrangian function value).}\label{fig:4-2}
\end{figure}

We conducted simulations of the distributed system identification over a total of 180 iteration rounds. All EDs and the server start with an initial guess $\boldsymbol{\theta}=\bf{0}$. For the 180 iteration rounds, the throughput and the relative information utility gain are shown in Fig. \ref{fig:4-1}(a) and Fig. \ref{fig:4-1}(b), respectively. The throughput is measured as the number of EDs that uploaded their primal and dual variables to the server. The observed policy trends align with the application of data-driven decision making. Here, the information utility gain under the hybrid policy is close to that of the channel-based policy.
 
The goal of the CPS operator is to identify the true system matrix $\boldsymbol{\theta}_{\mathrm{true}}$. In each round, we compute the value of the augmented Lagrangian function and relative gap between the server-estimated matrix $\boldsymbol{\theta}_{0}$ and the true system matrix $\boldsymbol{\theta}_{\mathrm{true}}$. The result is shown in Fig. \ref{fig:4-2}(a)-(b). It's evident that the hybrid approach contributes to reductions in both the Lagrangian function and the optimality gap. However, concerning the gap between the estimated matrix and the real matrix, the utility-based policy does not perform so well compared to the channel-based approach. In this case, the importance of more frequent updates could be more important. The consensus update will be slow if the central node estimation cannot gather enough information.  Compared to channel-based policy, hybrid policy can reduce approximately $30\%$ of iteration time to reach the relative optimality gap $10^{-3}$ within the same communication budget.

\section{Conclusion} \label{conclusion}
In this study, we presented a novel goal-oriented RBA framework tailored for future CPSs, termed as `goal-oriented communications'. Unlike traditional communication resource allocation methods, this approach prioritizes information transmission essential to CPS goals over purely transmission-centric goals. This paradigm is poised to gain prominence, especially as communication devices evolve to predominantly encompass massive machines and things, rather than humans. We conducted simulations encompassing typical use cases of CPSs. These simulations validate the efficacy of our proposed goal-oriented RBA strategy, which considers both the information semantics/ utility and the channel conditions.

Future works can be extended to incorporate a more detailed channel model, possibly including small-fading effects. A more detailed and realistic model can help to further release the potential of the goal-oriented communication paradigm. Future research could also explore more applications such as transient control of digital twin systems in CPSs, where the requirements for delay and reliability are more stringent. Furthermore, if allowed, a co-design framework for CPS operation and communication can be developed, which can further increase the system's efficiency.

\bibliographystyle{IEEEtran}
\bibliography{IEEEabrv,Reference}

\section*{Appendix}


\subsection{Proof for Proposition \ref{distributed}}
Before proving proposition \ref{distributed}, we first provide the lemmas below:

\begin{assume}
     The individual loss function $L_j(\mathsf{z}_j;\boldsymbol{\theta }_j)$ is convex and $\kappa_j$-smooth.
\end{assume}

\begin{definition}
    A differentiable function $f(\boldsymbol{x})$ is called $\rho$ strongly-convex if $\nabla^2 f(\boldsymbol{x}) \succeq \rho\mathbf{I}$. For $\rho$ strongly convex functions, the following is true:
    \begin{equation}
        f(\boldsymbol{x}') \geq f(\boldsymbol{x})+\langle\nabla f(\boldsymbol{x}), \boldsymbol{x}'-\boldsymbol{x}\rangle+\frac{\rho}{2}\|\boldsymbol{x}'-\boldsymbol{x}\|_2^2, \forall \boldsymbol{x}',\boldsymbol{x}.
    \end{equation}
\end{definition}

\begin{lemma}
    The augmented Lagrangian function $\mathcal{L} \left( \boldsymbol{\theta }_0,\boldsymbol{\theta }_j,\boldsymbol{\lambda }_j \right) $ is $\rho$ strongly convex with respect to $\boldsymbol{\theta}_0$ and $\boldsymbol{\theta}_j$, respectively.
\end{lemma}
\begin{proof}
    It is easy to verify that:
    \begin{equation}
    \nabla _{\boldsymbol{\theta }_j}^{2}\mathcal{L} =\nabla _{\boldsymbol{\theta }_j}^{2}L_j+\rho \mathbf{I} \succeq \rho \mathbf{I},~~\nabla _{\boldsymbol{\theta }_0}^{2}\mathcal{L} =\rho \mathbf{I} \succeq \rho \mathbf{I},
    \end{equation}
    where $\nabla _{\boldsymbol{\theta }_j}^{2}L_j\succeq \mathbf{0}$ because $L_j(\mathsf{z}_j;\boldsymbol{\theta }_j)$ is convex.
\end{proof}

\begin{lemma} \label{dual}
  The dual variable updates for two consecutive rounds can be bounded as follows:
    \begin{equation}
        \left\| \boldsymbol{\lambda }_{j,k+1}-\boldsymbol{\lambda }_{j,k} \right\| _2\leqslant \kappa _j\left\| \boldsymbol{\theta }_{j,k+1}-\boldsymbol{\theta }_{j,k} \right\| _2.
    \end{equation}
\end{lemma}

\begin{proof}
    For EDs that take part in updates in round $s_j=1$.  $\boldsymbol{\vartheta }_{j,k+1}$ is the solution to problem \eqref{eq:step2} and thus it satisfies the following optimality condition:
    \begin{equation}
        \nabla _{\boldsymbol{\theta }_j}L_j\mid_{\boldsymbol{\theta }_{j,k+1}}^{}+\boldsymbol{\lambda }_{j,k}+\rho \left( \boldsymbol{\theta }_{j,k+1}-\boldsymbol{\theta }_{0,k+1} \right) =0.
    \end{equation}
    Combing the dual update step \eqref{eq:step3}, we get:
    \begin{equation}
    \nabla _{\boldsymbol{\theta }_j}L_j\mid_{\boldsymbol{\theta }_{j,k+1}}^{}=-\boldsymbol{\lambda }_{j,k+1}.
    \end{equation}    
    Thereby, we have:
    \begin{equation}
    \begin{aligned}
      \left\| \boldsymbol{\lambda }_{j,k+1}-\boldsymbol{\lambda }_{j,k} \right\| _2&=\left\| \nabla _{\boldsymbol{\theta }_j}L_j\mid_{\boldsymbol{\theta }_{j,k+1}}^{}-\nabla _{\boldsymbol{\theta }_j}L_j\mid_{\boldsymbol{\theta }_{j,k}}^{} \right\| _2
      \\
      &\leqslant \kappa _j\left\| \boldsymbol{\theta }_{j,k+1}-\boldsymbol{\theta }_{j,k} \right\| _2,     
    \end{aligned}
    \end{equation}        
     where the inequality stems from the $\kappa_j$-smooth condition.\\
     For $j\notin \{ j:s_j=1 \}$, the inequality is trivial because both sides are zero.
\end{proof}

Now, we are ready to prove proposition \ref{distributed}.

\begin{proof}
    The changes in Lagrangian function are decomposed into three parts:
    \begin{equation}
        \begin{aligned}
        &\mathcal{L} \left( \boldsymbol{\theta }_{0,k+1},\boldsymbol{\theta }_{j,k+1},\boldsymbol{\lambda }_{j,k+1} \right) -\mathcal{L} \left( \boldsymbol{\theta }_{0,k},\boldsymbol{\theta }_{j,k},\boldsymbol{\lambda }_{j,k} \right) 
        \\
        =&\mathcal{L} \left( \boldsymbol{\theta }_{0,k+1},\boldsymbol{\theta }_{j,k+1},\boldsymbol{\lambda }_{j,k+1} \right) -\mathcal{L} \left( \boldsymbol{\theta }_{0,k+1},\boldsymbol{\theta }_{j,k+1},\boldsymbol{\lambda }_{j,k} \right) 
        \\
        +&\mathcal{L} \left( \boldsymbol{\theta }_{0,k+1},\boldsymbol{\theta }_{j,k+1},\boldsymbol{\lambda }_{j,k} \right) -\mathcal{L} \left( \boldsymbol{\theta }_{0,k+1},\boldsymbol{\theta }_{j,k},\boldsymbol{\lambda }_{j,k} \right) 
        \\
        +&\mathcal{L} \left( \boldsymbol{\theta }_{0,k+1},\boldsymbol{\theta }_{j,k},\boldsymbol{\lambda }_{j,k} \right) -\mathcal{L} \left( \boldsymbol{\theta }_{0,k},\boldsymbol{\theta }_{j,k},\boldsymbol{\lambda }_{j,k} \right) .  
        \end{aligned}
    \end{equation}
    The first part can be computed as follows:
    \begin{equation} \label{eq:part1}
        \begin{aligned}
        &\mathcal{L} \left( \boldsymbol{\theta }_{0,k+1},\boldsymbol{\theta }_{j,k+1},\boldsymbol{\lambda }_{j,k+1} \right) -\mathcal{L} \left( \boldsymbol{\theta }_{0,k+1},\boldsymbol{\theta }_{j,k+1},\boldsymbol{\lambda }_{j,k} \right) 
        \\
        \overset{\left( \mathrm{I} \right)}{=}&\sum\nolimits_j^{}{\left< \boldsymbol{\lambda }_{j,k+1}-\boldsymbol{\lambda }_{j,k},\boldsymbol{\theta }_{j,k+1}-\boldsymbol{\theta }_{0,k+1} \right>}
        \\
        \overset{\left( \mathrm{II} \right)}{=}&\sum\nolimits_{j\in \{ j:s_j=1 \} }^{}{\frac{1}{\rho}\left\| \boldsymbol{\lambda }_{j,k+1}-    \boldsymbol{\lambda }_{j,k} \right\| _{2}^{2}}.
         \end{aligned}
    \end{equation}
    In $(\mathrm{I})$, the expression of the Lagrangian function is used. In $(\mathrm{II})$, the update law of dual multipliers step \eqref{eq:step3} is used. Besides, for $j\notin \{ j:s_j=1 \}$,  $\boldsymbol{\lambda }_{j,k+1}-    \boldsymbol{\lambda }_{j,k}=0$.
    
    The second part can be bounded as follows:
    \begin{equation}\label{eq:part2}
        \begin{aligned}
	&\mathcal{L} \left( \boldsymbol{\theta }_{0,k+1},\boldsymbol{\theta }_{j,k+1},\boldsymbol{\lambda }_{j,k} \right) -\mathcal{L} \left( \boldsymbol{\theta }_{0,k+1},\boldsymbol{\theta }_{j,k},\boldsymbol{\lambda }_{j,k} \right)
        \\
	\overset{(\mathrm{I})}{\leqslant}&\sum_j{\left. \langle \left. \nabla _{\boldsymbol{\theta }_j}\mathcal{L} \right|_{\boldsymbol{\theta }_{j,k+1}},\boldsymbol{\theta }_{j,k+1}-\boldsymbol{\theta }_{j,k} \right.}\rangle -\frac{\rho}{2}\left\| \boldsymbol{\theta }_{j,k+1}-\boldsymbol{\theta }_{j,k} \right\| _{2}^{2}
        \\
	\overset{\left( \mathrm{II} \right)}{\leqslant}&\sum_{j\in \{ j:s_j=1 \} }{\left. \langle \left. \nabla _{\boldsymbol{\theta }_j}\mathcal{L} \right|_{\boldsymbol{\theta }_{j,k+1}},\boldsymbol{\theta }_{j,k+1}-\boldsymbol{\theta }_{j,k} \right.}\rangle -\frac{\rho}{2}\left\| \boldsymbol{\theta }_{j,k+1}-\boldsymbol{\theta }_{j,k} \right\| _{2}^{2}
        \\
	\overset{\left( \mathrm{III} \right)}{\leqslant}&-\sum_{j\in \{ j:s_j=1 \}}{\frac{\rho}{2}}\left\| \boldsymbol{\theta }_{j,k+1}-\boldsymbol{\theta }_{j,k} \right\| _{2}^{2}.
        \end{aligned}
    \end{equation}
     In $(\mathrm{I})$, the property of strongly convex is used when $\boldsymbol{\theta }_j =\boldsymbol{\theta }_{j,k+1}$. In $(\mathrm{II})$, we use the fact that $\boldsymbol{\theta }_{j,k+1}-\boldsymbol{\theta }_{j,k}=0$ for $j \notin \{ j:s_j=1 \}$. In $(\mathrm{III})$, the optimal condition of the decomposed sub-problem \eqref{eq:step2} is used.

    The third part can be bounded as follows:
    \begin{equation}\label{eq:part3}
        \begin{aligned}
         &\mathcal{L} \left( \boldsymbol{\theta }_{0,k+1},\boldsymbol{\theta }_{j,k},\boldsymbol{\lambda }_{j,k} \right) -\mathcal{L} \left( \boldsymbol{\theta }_{0,k},\boldsymbol{\theta }_{j,k},\boldsymbol{\lambda }_{j,k} \right) 
        \\
        \overset{\left( \mathrm{I} \right)}{\leqslant} &\left. \langle \left. \nabla _{\boldsymbol{\theta }_0}\mathcal{L} \right|_{\boldsymbol{\theta }_{0,k+1}},\boldsymbol{\theta }_{0,k+1}-\boldsymbol{\theta }_{0,k} \right. \rangle -\frac{\rho}{2}\left\| \boldsymbol{\theta }_{0,k+1}-\boldsymbol{\theta }_{0,k} \right\| _{2}^{2}
        \\
        \overset{\left( \mathrm{II} \right)}{\leqslant} &-\frac{\rho}{2}\left\| \boldsymbol{\theta }_{0,k+1}-\boldsymbol{\theta }_{0,k} \right\| _{2}^{2}.
        \end{aligned}
    \end{equation}
     In $(\mathrm{I})$, the property of strongly convex is used when $\boldsymbol{\theta }_0 =\boldsymbol{\theta }_{0,k+1}$. In $(\mathrm{II})$, the optimal condition of the decomposed sub-problem \eqref{eq:step1} is used.

    Combining the three parts above:
    \begin{equation}
        \begin{aligned}
        &\mathcal{L} \left( \boldsymbol{\theta }_{0,k+1},\boldsymbol{\theta }_{j,k+1},\boldsymbol{\lambda }_{j,k+1} \right) -\mathcal{L} \left( \boldsymbol{\theta }_{0,k},\boldsymbol{\theta }_{j,k},\boldsymbol{\lambda }_{j,k} \right) 
        \\
        \leqslant &-\sum_{j\in \{ j:s_j=1 \} }^{}{\frac{\rho}{2}\left\| \boldsymbol{\theta }_{j,k+1}-\boldsymbol{\theta }_{j,k} \right\| _{2}^{2}}-\frac{\rho}{2}\left\| \boldsymbol{\theta }_{0,k+1}-\boldsymbol{\theta }_{0,k} \right\| _{2}^{2}
        \\
        \,\,  &+\sum_{j\in \{ j:s_j=1 \} }^{}{\frac{1}{\rho}\left\| \boldsymbol{\lambda }_{j,k+1}-\boldsymbol{\lambda }_{j,k} \right\| _{2}^{2}}
        \\
        \overset{\left(\mathrm{I}\right)}{\leqslant}&-\sum\nolimits_{j\in \{ j:s_j=1 \}}^{}{\left( \frac{\rho}{2}-\frac{\kappa _j}{\rho} \right) \left\| \boldsymbol{\theta }_{j,k+1}-\boldsymbol{\theta }_{j,k} \right\| _{2}^{2}}
        \\
        &-\frac{\rho}{2}\left\| \boldsymbol{\theta }_{0,k+1}-\boldsymbol{\theta }_{0,k} \right\| _{2}^{2}.
        \end{aligned}
    \end{equation}
    where in (I), we use lemma \ref{dual}. Thereby, when $\rho$ is sufficiently large, we can guarantee the Lagrangian is always reducing. Replace the coefficients and we get proposition 4.
\end{proof}

\begin{IEEEbiography}[{\includegraphics[width=1in,height=1.25in,clip,keepaspectratio]{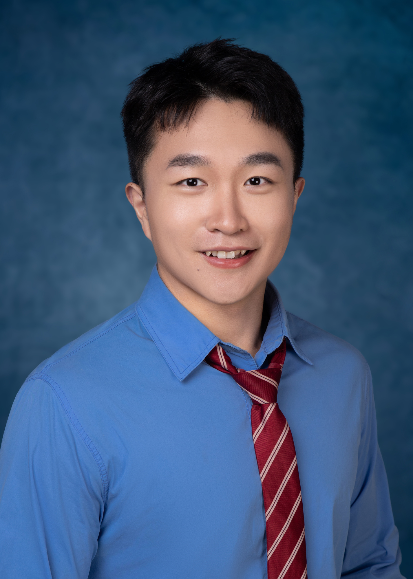}}]{Cheng Feng} received the B.S. degree in Electrical Engineering in Huazhong University of Science and Technology in June, 2019, and the Ph.D. degree in Electrical Engineering from Tsinghua University in June, 2024. During February 2023 to August 2023, he was a visiting scholar in Automatic Control Lab (ifA), ETH Zurich. He is the incoming Ezra Postdoctoral Associate in Cornell University. His research interests include cyber-physical systems data analytics in energy systems.
\end{IEEEbiography}

\begin{IEEEbiography}[{\includegraphics[width=1in,height=1.264in,clip,keepaspectratio]{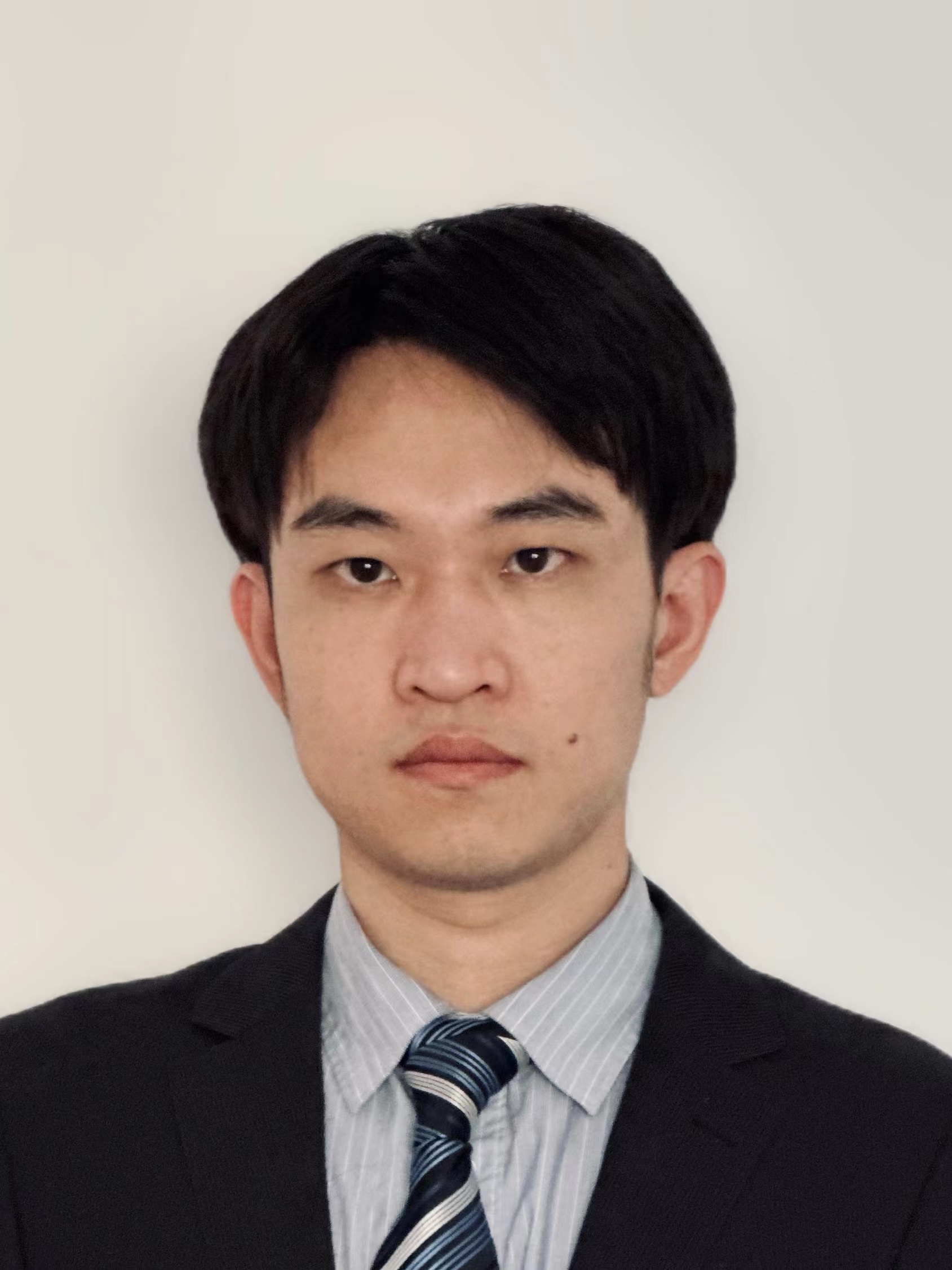}}]{Kedi Zheng} received the B.S. and Ph.D. degrees in electrical engineering from Tsinghua University, Beijing, China, in 2017 and 2022, respectively. He is currently a Research Associate with Tsinghua University. His research interests include data analytics in power systems and electricity markets.
\end{IEEEbiography}

\begin{IEEEbiography}[{\includegraphics[width=1in,height=1.264in,clip,keepaspectratio]{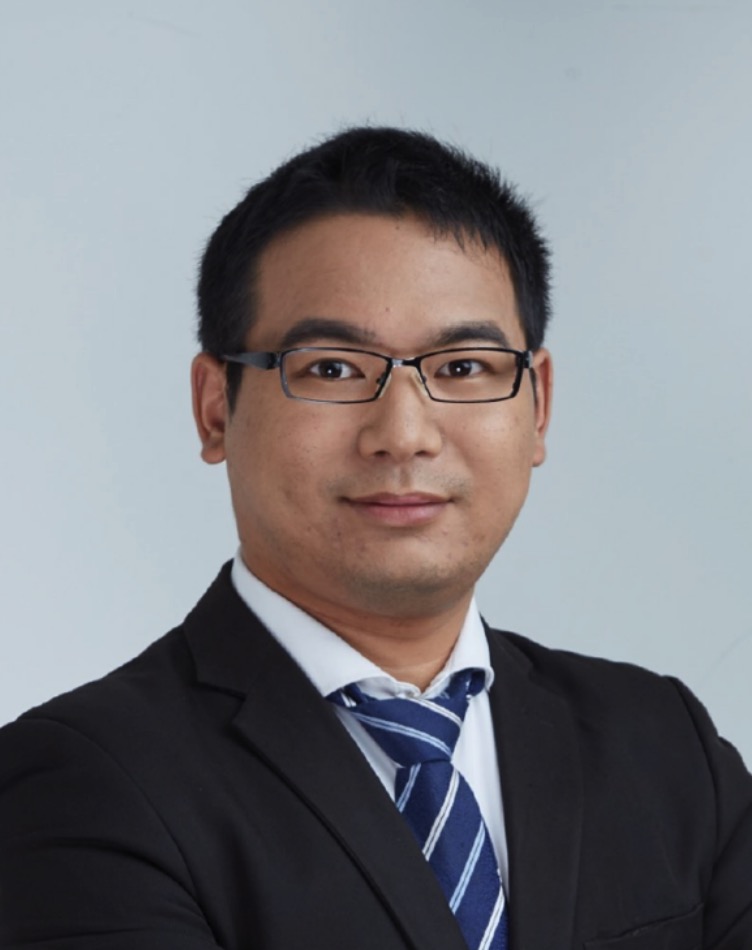}}]{Yi Wang} received the B.S. degree from Huazhong University of Science and Technology in June 2014, and the Ph.D. degree from Tsinghua University in January 2019. He was a visiting student with the University of Washington from March 2017 to April 2018. He served as a Postdoctoral Researcher in the Power Systems Laboratory, at ETH Zurich from February 2019 to August 2021. He is currently an Assistant Professor with the Department of Electrical and Electronic Engineering, The University of Hong Kong. His research interests include data analytics in smart grids, energy forecasting, multi-energy systems, Internet of-things, and cyber-physical-social energy systems.
\end{IEEEbiography}

\begin{IEEEbiography}[{\includegraphics[width=1in,height=1.264in,clip,keepaspectratio]{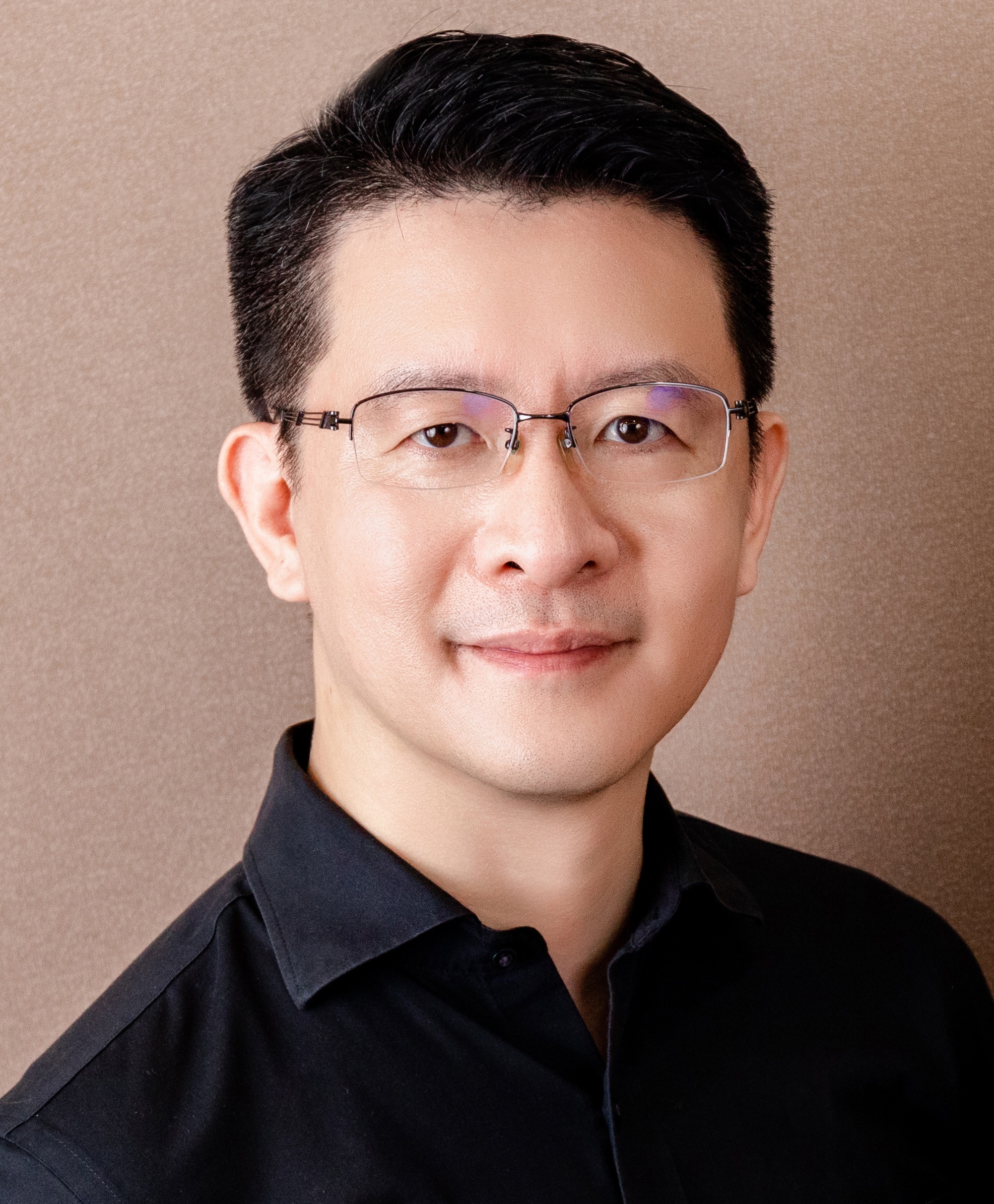}}]{Kaibin Huang} received the B.Eng. and M.Eng. degrees from the National University of Singapore and the Ph.D. degree from The University of Texas at Austin, all in electrical engineering. He is a Professor and the Head at the Dept. of Electrical and Electronic Engineering, The University of Hong Kong (HKU), Hong Kong. He received the IEEE Communication Society’s 2021 Best Survey Paper, 2019 Best Tutorial Paper, 2019 and 2023 Asia–Pacific Outstanding Paper, 2015 Asia–Pacific Best Paper Award, and the best paper awards at IEEE GLOBECOM 2006 and IEEE/CIC ICCC 2018. He has been named as a Web-of-Science Highly Cited Researcher in 2019-2024 and an AI 2000 Most Influential Scholar (Top 30 in Internet of Things) in 2023-2024. He was a IEEE Distinguished Lecturer of both the IEEE Communications Society and the IEEE Vehicular Technology Society. He is a member of the Engineering Panel of Hong Kong Research Grants Council (RGC) and a RGC Research Fellow (2021 Class). He received the Outstanding Teaching Award from Yonsei University, South Korea, in 2011. He is an Area Editor of IEEE Transactions on Wireless Communications, IEEE Transactions on Machine Learning in Communications and Networking, and IEEE Transactions on Green Communications and Networking. Previously, he served on the Editorial Boards for IEEE Journal on Selected Areas in Communications (JSAC) and IEEE Wireless Communication Letters. He has guest edited special issues of IEEE JSAC, IEEE Journal of Selected Areas in Signal Processing, and IEEE Communications Magazine, and IEEE Network. He served as the Lead Chair for the Wireless Communications Symposium of IEEE Globecom 2017 and the Communication Theory Symposium of IEEE GLOBECOM 2023 and 2014, and the TPC Co-chair for IEEE PIMRC 2017 and IEEE CTW 2023 and 2013. He is the founding President of the HKU chapter of National Academy of Inventors.
\end{IEEEbiography}

\begin{IEEEbiography}[{\includegraphics[width=1in,height=1.264in,clip,keepaspectratio]{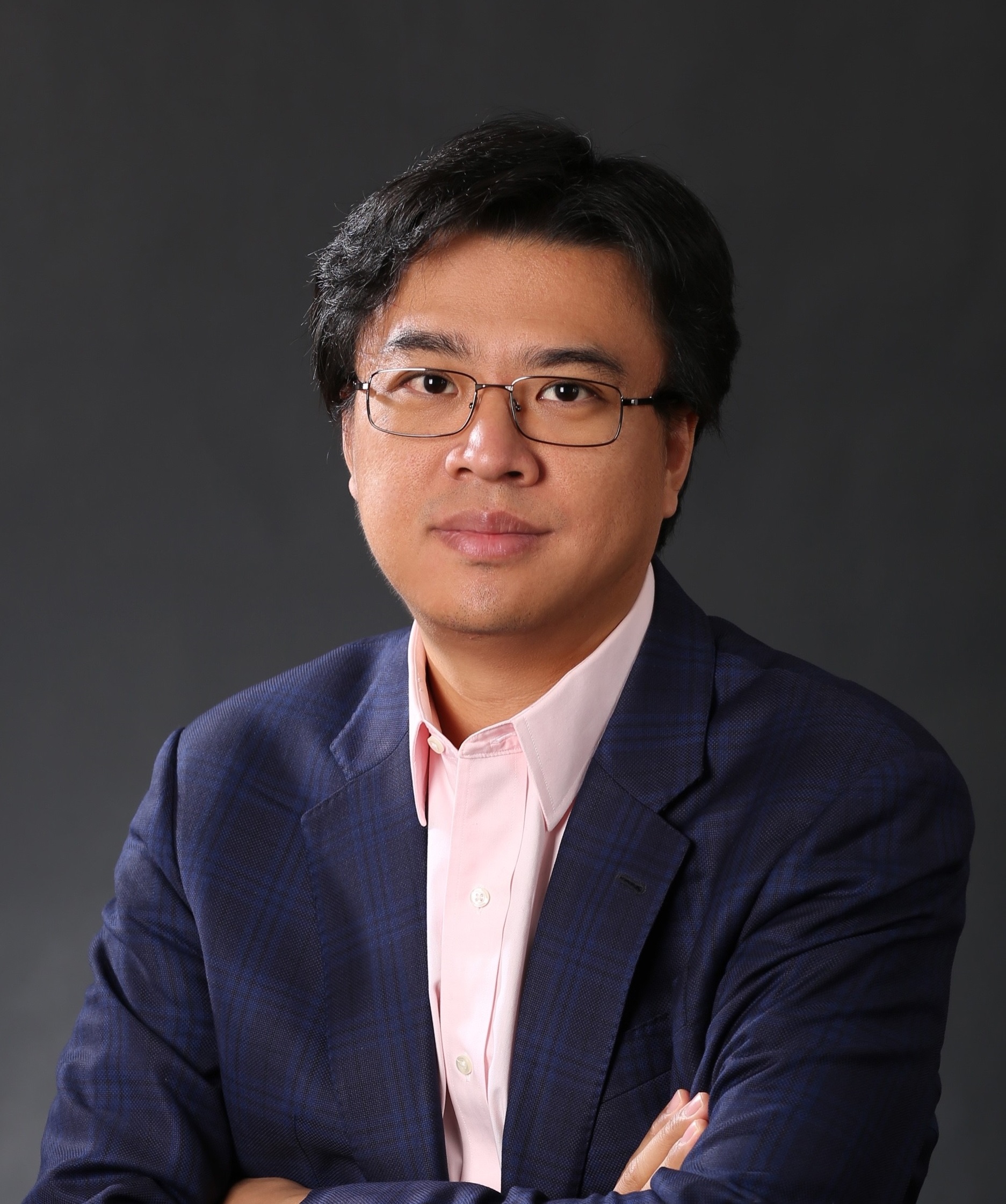}}]{Qixin Chen} received the Ph.D. degree from the Department of Electrical Engineering, Tsinghua University, Beijing, China, in 2010. He is currently a tenured professor with Tsinghua University. His research interests include electricity markets, power system economics and optimization, low-carbon electricity, and data analytics in power systems.
\end{IEEEbiography}
\ifCLASSOPTIONcaptionsoff
  \newpage
\fi


\end{document}